\documentclass{article}
\date{
October 31, 2014%
}

\usepackage[a4paper]{geometry}

\advance\textheight by 2cm
\advance\topmargin -1cm

\advance\textwidth by 2cm
\advance\oddsidemargin by -1cm
\advance\evensidemargin by -1cm

\usepackage{times}
\usepackage[scaled=0.88]{helvet}

\usepackage[T1]{fontenc}
\usepackage{pdfsync} 
\usepackage{graphicx}
\usepackage{amsfonts}
\usepackage{bm,amsmath}
\usepackage{mathrsfs,theorem}
\usepackage{algorithm,algorithmic}
\floatname{algorithm}{Algorithm}
\usepackage[usenames,dvipsnames]{color}
\usepackage{colortbl} 
\usepackage{graphicx}
\usepackage{mathrsfs}
\usepackage{booktabs}
\usepackage{xspace}
\usepackage{url}
\usepackage{rotating}

\newcounter{noqed}
\newcommand{\qed}{ \ifmmode\mbox{ }\fi\rule[-.05em]{.3em}{.7em}\setcounter{noqed}{0}}
\newenvironment{proof}[1][{}]{\noindent{\bf Proof#1. }\setcounter{noqed}{1}}{\ifnum\value{noqed}=1\qed\fi\par\medskip}

\floatname{algorithm}{Algorithm}

\newfloat{Algorithm}{htpb}{alg}
\newcounter{prgline}
\newcommand{\tauh}{\tau_{\mathrm h}}
\newcommand{\pl}{\theprgline\addtocounter{prgline}{1}}
\renewcommand{\epsilon}{\varepsilon}
\newcommand{\ww}{w}

\def\..{\,\mathpunct{\ldotp\ldotp}} 

\newcommand{\IF}{\text{\textbf{if}}\xspace}
\newcommand{\FI}{\text{\textbf{fi}}\xspace}

\newcommand{\FOR}{\text{\textbf{for}}\xspace}

\newcommand{\END}{\text{\textbf{end}}\xspace}
\newcommand{\DO}{\text{\textbf{do}}\xspace}
\newcommand{\OD}{\text{\textbf{od}}\xspace}

\newcommand{\THEN}{\text{\textbf{then}}\xspace}
\newcommand{\ELSE}{\text{\textbf{else}}\xspace}

\newcommand{\RETURN}{\text{\textbf{return}}\xspace}

\newcommand{\LAND}{\text{\textbf{and}}\xspace}

\newcommand{\WHILE}{\text{\textbf{while}}\xspace}

\newcommand{\FUNCTION}{\text{\textbf{function}}\xspace}

\newcommand{\pp}{\raisebox{.08em}{\text{\small++}}}

\newtheorem{theorem}{Theorem}

\newcommand{\signum}{\operatorname{sgn}}
\newcommand{\br}{\bm r}
\newcommand{\bs}{\bm s}
\newcommand{\R}{\mathbf R}
\newcommand{\la}{\langle}
\newcommand{\ra}{\rangle}

\title{A Weighted Correlation Index for Rankings with Ties}

\author{%
Sebastiano Vigna\thanks{Sebastiano Vigna has been supported by the
EU-FET grant NADINE (GA 288956).}\\
Universit\`a degli Studi di Milano, Italy}

\begin{document}
\clubpenalty=10000 
\widowpenalty = 10000
\bibliographystyle{plain}
\maketitle

\begin{abstract}
Understanding the correlation between two different scores for the same set of items
is a common problem in graph analysis and information retrieval. The most commonly used
statistics that quantifies this correlation is Kendall's $\tau$; however, the
standard definition fails to capture that discordances between items with high
rank are more important than those between items with low rank. Recently,
a new measure of correlation based on~\emph{average precision} has been 
proposed to solve this problem, but like many alternative proposals in
the literature it assumes that there are \emph{no ties} in the scores. This is a major deficiency
in a number of contexts, and in particular while comparing centrality scores on large graphs,
as the obvious baseline, indegree, has a very large number of ties in social networks and web graphs. We propose to
extend Kendall's definition in a natural way to take into account weights in the
presence of ties. We prove a number of interesting mathematical properties of
our generalization and describe
an $O(n\log n)$ algorithm for its computation. We also validate the usefulness
of our weighted measure of correlation using experimental data on social networks and web graphs.
\end{abstract}

%
%


\section{Introduction}

In information retrieval, one is often faced with different scores\footnote{We
purposely and consistently use ``score'' to denote real numbers associated to
items, and ``rank'' to denote ordinal positions. The two terms are used somewhat
interchangeably in the literature, but in this paper the distinction is
important as we assume that scores of different items can be identical.} for the
same set of items.
This includes the lists of documents returned by different search engines and their
associated relevance scores, the lists of query recommendation returned by different algorithms, and also the
score associated to each node of a graph by different centrality measures (e.g.,
indegree and Bavelas's closeness~\cite{BavCPTOG}).

In most of the literature, the scores are assumed to be without ties, thus
inducing a \emph{ranking} of the elements. At that point, correlation statistics
such as Spearman's rank correlation coefficient~\cite{SpePMATT} and Kendall's
$\tau$~\cite{KenNMRC} can be used to evaluate the similarity of the rankings.
Spearman's correlation coefficient is equivalent to the traditional linear
correlation coefficient computed on ranks of items. Kendall's $\tau$, instead,
is proportional to the number of pairwise adjacent swaps needed to convert one
ranking into the other.

For a number of reasons, Kendall's $\tau$ has become a standard statistic to
compare the correlation between two ranked lists. Such reasons include fast computation
($O(n \log n)$, where $n$ is the length of the list, using Knight's algorithm~\cite{KnCMCKTUD}),
and the existence of a variant that takes care of ties~\cite{KenTTRP}.

The explicit treatment of ties is of great importance when comparing global
\emph{exogenous} relevance scores in large collections of web documents. The
baseline of such scores is indegree---the number of documents containing
hypertextual link to a given document. More sophisticated approaches include
Katz's index~\cite{KatNSIDSA}, PageRank~\cite{PBMPCR}, and countless variants.
Due to the highly skewed indegree distribution, a very large number of documents
share the same indegree, and the same happens of many other scores: it is thus
of uttermost importance that the evaluation of correlation takes into account
ties as first-class citizens.

On the other hand, Kendall's $\tau$ has some known problems that motivated the
introduction of several weighted variants. In particular, a striking difference
often emerges between the anecdotal evidence of the top elements by different scores
being almost identical, and the $\tau$ value being quite low. This is due
to a known phenomenon: the scores of important items tend to be highly correlated
in all reasonable rankings, whereas most of the remaining items are ranked in slightly different
ways, introducing a large amount of noise, yielding a low $\tau$ value.

This problem motivates the definition of correlation statistics that consider
more important correlation between highly ranked items. In particular, recently Yilmaz, Aslam
and Robertson introduced a statistics, named \emph{AP (average precision)
correlation}~\cite{YARNRCCIR}, which aims at considering more important swaps
between highly ranked items. The need for such a measure is very well motivated in the
introduction of their paper, and we will not repeat here their detailed
discussion.

In this paper, we aim at providing a measure of correlation in the
same spirit of the definition oh Yilmaz, Aslam and Robertson, but taking smoothly ties into account. We will actually
define a general notion of weighting for Kendall's $\tau$, and develop its mathematical properties.
Since it is important that such a statistics is computable on very large data
sets, we will provide a generalization of Knight's algorithm that can be applied
whenever the weighting depends additively or multiplicatively on a weight
assigned to each item. The same algorithm can be used to compute AP correlation
in time $O(n\log n)$.

All data and software used in this paper are available as part of the LAW
software library under the GNU General Public
License.\footnote{\url{http://law.di.unimi.it/}}

\section{Related work}

Shieh~\cite{ShiWKTS} wrote the one of the first papers proposing a generic weighting of Kendall's $\tau$. She
assumes from the very start that there are no ties, and assign to the exchange between $i$ and $j$ a weight $w_{ij}$.
Her motivation is the \emph{fidelity evaluation of software packages for structural engineering}, in which a set 
of variables is ranked in two different ways, and one would like to emphasize agreement on the most important ones.
In particular, she concentrates on weights given by the product of two weights associated with the elements participating
in the exchange. Our work can be seen as a generalization of her approach,
albeit we combine weights differently.

Kumar and Vassilvitskii~\cite{KuVGDBR} study a definition that extends Shieh's
taking into account \emph{position weights} and \emph{similarity between
elements}. Again, they assume that ties are broken arbitrarily, which is an
unacceptable assumption if large sets of elements have the same score. Fagin,
Kumar and Sivakumar~\cite{FKSCTKL} use instead \emph{penalty weights} to apply Kendall's $\tau$ just to
the top $k$ elements of two ranked lists (with no ties). Exchanges partially or completely outside the
top $k$ elements obtain different weights.

Finally, the recent quoted work of Yilmaz, Aslam and Robertson~\cite{YARNRCCIR} on AP correlation
is the closest to ours in motivation and methodology, albeit targeted at ranked lists 
with no ties.

We remark that analogous research exists in association with Spearman's
correlation: Iman and Conover~\cite{ImCMTDC}, for example, study the usage of
\emph{Savage scores}~\cite{SavCTROS} instead of ranks when comparing ranked lists.
Savage scores for a ranked list of $n$ elements are given by $\sum_{j=i}^n 1/j$,
where $i$ is the rank (starting at one) of an element. Spearman's correlation
applied to Savage scores considers more important elements at the top of a
ranked list.

Recently, Webber, Moffat and Zobel~\cite{WMZSMIR} have described a similarity
measure for \emph{indefinite rankings}---rankings that might have different
lengths and contain different elements. Their work has some superficial
resemblance with the approach of~\cite{KuVGDBR,YARNRCCIR} and our work, as it
give preeminence to differences at the top of ranked lists, but it is not
technically a correlation index, as it is based on measuring overlaps of
infinite lists, rather than on exchanges. Thus, the basic condition for a
correlation index (i.e., that inverting the list one obtains the minimum
possible correlation, usually standardized to $-1$), is not even expressible in
their framework.
Moreover, their measure, being defined on infinite lists, needs the
fundamental assumption that the weight function applied to overlaps 
must be \emph{summable}; in particular, they make importance decrease
exponentially.
As we will discuss in Section~\ref{sec:choosing}, and verify experimentally in
Section~\ref{sec:experiments}, such a choice is a reasonable framework for very
short lists, or when only very first elements are relevant (e.g., because one
is modeling user behavior), but it would completely flatten the results of our
correlation index on large examples, depriving it from its discriminatory power,
even if the weight function would decrease just quadratically.

A fascinating proposal, entirely orthogonal to the ones we discussed, is the
idea of weighting Kemeny's distance between permutations proposed by Farnoud and
Milenkovic~\cite{FaMAACDRCA}.
In this proposal, Kemeny's distance between two permutations $\pi$ and $\sigma$
is characterized as the minimum number of \emph{adjacent transpositions} (i.e.,
transpositions of the form $(i\,i+1)$) that turn $\pi$ into $\sigma$.
At this point, one can define a \emph{weight} associated to each adjacent
transposition, and by assigning larger weights to adjacent transpositions
with smaller indices one can make differences in the top part of the
permutations more important than differences in the bottom part. The right
notion of weighted distance turns out to be the minimum sum of weights of a
sequence of adjacent transposition that turn $\pi$ into $\sigma$. The
interesting property of this approach is that avoids the need for a \emph{ground
truth} (an intrinsic notion of importance of an element), which is necessary,
implicitly or explicitly, to weigh an exchange in the approaches
of~\cite{ShiWKTS,YARNRCCIR} and in the one discussed in this paper.
The main drawbacks, presently, are that the weight assignment is not very
intuitive (as it is related to positions, rather than to elements) and that more
work is needed to extend the distance into a proper correlation index in the case of ties.

\section{Motivation}
\label{sec:motivation}

The need for weighted correlation measures in the case of ranked list has been
articulated in detail in previous work. Here we will focus on the case of
centrality measures for graphs. Consider the graph of English
Wikipedia\footnote{More precisely, a specific snapshot of Wikipedia that will be made 
public by the author. The graph does not contain template pages.}, which has about
four million nodes and one hundred million arcs.
In this graph, $99.95$\% of the nodes have the same indegree of some other
node---for example, more than a half million node has indegree one. It is
clearly mandatory, when
 computing the correlation of other scores with indegree, that ties are taken
 into consideration in a systematic way
(e.g., not broken arbitrarily).

We will consider four other commonly used scores based on the adjacency matrix
$A$ of the Wikipedia graph.
One is PageRank~\cite{PBMPCR}, which is defined by
\[
\mathbf 1/n \sum_{k\geq0} (\alpha\bar A)^k,
\]
where $\alpha\in[0\..1)$ is a \emph{damping factor} and $\bar A$ is a stochasticization of
$A$: every row not entirely made of zeroes is divided by its sum, so to have $\ell_1$ norm one.

The other index we consider is Katz's~\cite{KatNSIDSA}, which is defined by
\[
\mathbf 1 \sum_{k\geq0} (\alpha A)^k,
\]
where $\alpha\in[0\..1/\lambda)$ is an attenuation factor depending on
$\lambda$, the dominant eigenvalue of $A$~\cite{MeyMAALA}.
In both cases, we take $\alpha$ in the middle of the allowed interval (using
different values does not change the essence of what follows, unless they are
extreme).

A different kind of score is provided by Bavelas's \emph{closeness}. The closeness of $x$ is defined by 
\[ \frac1{\sum_{d(y,x)<\infty} d(y,x)},\]
where $d(-,-)$ denotes the usual graph distance. Note that we
have to eliminate nodes at infinite distance to avoid zeroing all scores. By
definition the closeness of a node with indegree zero is zero. Finally, we
consider \emph{harmonic centrality}~\cite{BoVAC}, a modified version of Bavelas's closeness designed for directed graphs that are not
strongly connected; the harmonic centrality of $x$ is defined by \[       
\sum_{y\neq x}\frac1{d(y,x)}.\] 

These scores provide an interesting mix: indegree is an obvious baseline, and
entirely local. PageRank and Katz are similar in their definition, but the
normalization applied to $A$ makes the scores quite different (at least in
theory).
Finally, closeness and harmonic centrality are of a completely
different nature, having no connection with dominant eigenvectors or Markov chains.

Our first empirical observation is that, looking just at the very top pages of
Wikipedia (Table~\ref{tab:topwiki}; entries in boldface are unique to the list
they belong to, here and in the following), we perceive these scores as almost
identical, except for closeness, which displays almost random values. The latter behavior is a known phenomenon:
nodes that are almost isolated obtain a very high closeness score (this is why
harmonic centrality was devised).
We note also that harmonic centrality has a slightly different slant,
as it is the only ranking including Latin, Europe, Russia and the Catholic Church in the top $20$.

The problem is that these facts are not reflected in any way in the values of
Kendall's $\tau$ shown in Table~\ref{tab:kendall}.
If we exclude closeness, with the exception of the correlation between indegree
and Katz, all other correlation value fail to surpass the $0.9$ threshold,
usually considered the threshold for considering two rankings
equivalent~\cite{VooEHRD}. Actually, they are below the threshold
$0.8$, under which we are supposed to see considerable changes. The correlation
of closeness with harmonic centrality, moreover, is even more
pathological:
it is the \emph{largest} correlation.

An obvious observation is that, maybe, the score is lowered by a large
discordance in the rest of the rankings.
Table~\ref{tab:topscientist} tries to verify this intuition by listing the top
pages that are associated with the Wordnet category ``scientist'' 
in the Yago2 ontology data~\cite{HSBY2}. These pages have
considerably lower score (their rank is below 300), yet the first three
rankings are almost identical. Harmonic centrality is still slightly
different (Linnaeus is absent, and actually ranks $21$), which tells us that the
Kendall's $\tau$ is not giving completely unreasonable data. Nonetheless,
closeness continues to provide apparently random results.

We have actually to delve deep into Wikipedia, beyond rank 100\,000 using the
category ``cocktail'' to see that, finally, things settle down
(Table~\ref{tab:topcocktail}). While closeness still displays a few quirks, the
rankings start to stabilize.

To understand what happens in the very low-rank region, in
Table~\ref{tab:kendall2} we provide Kendall's $\tau$ as in
Table~\ref{tab:kendall}, but \emph{restricting the computation to nodes of
indegree 1 and 2}. As it is immediately evident, after stabilization the
low-rank region is fraught with noise and all correlation values drop
significantly.

The very high correlation between closeness and harmonic centrality is,
actually, not strange: on the nodes reachable from giant connected component of
our Wikipedia snapshot ($89$\% of the nodes) they agree almost exactly, as closeness is the 
reciprocal of a denormalized \emph{arithmetic} mean, whereas harmonic centrality
is the reciprocal of a denormalized \emph{harmonic} mean~\cite{BoVAC}. Even if
the remaining $11$\% of the nodes is completely out of place, making closeness useless, Kendall's $\tau$ tells us that it should be interchangeable
with harmonic centrality.
At the same time, Kendall's $\tau$ tells us that indegree is very different from
PageRank, which again goes completely against our empirical evidence.

In the rest of the paper, we will try to approach in a systematic manner these
problems by defining a new weighted correlation index for scores with ties.
 
\begin{table}[t]
\centering
\scriptsize
\begin{tabular}{lllll}
\multicolumn{1}{c}{Indegree} & \multicolumn{1}{c}{PageRank} & \multicolumn{1}{c}{Katz} & \multicolumn{1}{c}{Harmonic} & \multicolumn{1}{c}{Closeness} \\
\hline
United States & United States & United States & United States & \textbf{Kharqan Rural District} \\
List of sovereign states & Animal & List of sovereign states & United Kingdom & \textbf{Talageh-ye Sofla} \\
Animal & List of sovereign states & United Kingdom & World War II & \textbf{Talageh-ye Olya} \\
England & France & France & France & \textbf{Greatest Remix Hits (Whigfield album)} \\
France & Germany & Animal & Germany & \textbf{Suzhou HSR New Town} \\
Association football & Association football & World War II & Association football & \textbf{Suzhou Lakeside New City} \\
United Kingdom & England & England & English language & \textbf{Mepirodipine} \\
Germany & India & Association football & China & \textbf{List of MPs  \ldots  M--N} \\
Canada & United Kingdom & Germany & Canada & \textbf{List of MPs  \ldots  O--R} \\
World War II & Canada & Canada & India & \textbf{List of MPs  \ldots  S--T} \\
India & Arthropod & India & \textbf{Latin} & \textbf{List of MPs  \ldots  U--Z} \\
Australia & Insect & Australia & World War I & \textbf{List of MPs  \ldots  J--L} \\
London & World War II & London & England & \textbf{List of MPs  \ldots  C} \\
Japan & Japan & Italy & Italy & \textbf{List of MPs  \ldots  F--I} \\
Italy & Australia & Japan & \textbf{Russia} & \textbf{List of MPs  \ldots  A--B} \\
Arthropod & Village & New York City & \textbf{Europe} & \textbf{List of MPs  \ldots  D--E} \\
Insect & Italy & English language & Australia & \textbf{Esmaili-ye Sofla} \\
New York City & Poland & China & \textbf{European Union} & \textbf{Esmaili-ye Olya} \\
English language & English language & Poland & \textbf{Catholic Church} & \textbf{Levels of organization (ecology)} \\
Village & \textbf{National Reg. of Hist. Places} & World War I &
London & \textbf{Jacques Moeschal (architect)} \\
\end{tabular}
\caption{\label{tab:topwiki}Top $20$ pages of the English version of Wikipedia
following five different centrality measures.}
\end{table}

\begin{table}[t]
\centering
\scriptsize
\begin{tabular}{lllll}
\multicolumn{1}{c}{Indegree} & \multicolumn{1}{c}{PageRank} & \multicolumn{1}{c}{Katz} & \multicolumn{1}{c}{Harmonic} & \multicolumn{1}{c}{Closeness} \\
\hline
Carl Linnaeus & Carl Linnaeus & Carl Linnaeus & Aristotle & \textbf{No\"el
Bernard (botanist)} \\
Aristotle & Aristotle & Aristotle & Albert Einstein & \textbf{Charles Coquelin} \\
Thomas Jefferson & Thomas Jefferson & Thomas Jefferson & Thomas Jefferson & \textbf{Markku Kivinen} \\
Margaret Thatcher & Charles Darwin & Albert Einstein & Charles Darwin & \textbf{Angiolo Maria Colomboni} \\
Plato & Plato & Charles Darwin & Thomas Edison & \textbf{Om Prakash (historian)} \\
Charles Darwin & Albert Einstein & Karl Marx & \textbf{Alexander Graham Bell} & \textbf{Michel Mandjes} \\
Karl Marx & Karl Marx & Plato & \textbf{Nikola Tesla} & \textbf{Kees Posthumus} \\
Albert Einstein & Pliny the Elder & Margaret Thatcher & \textbf{William James} & \textbf{F. Wolfgang Schnell} \\
Vladimir Lenin & Vladimir Lenin & Vladimir Lenin & Isaac Newton & \textbf{Christof Ebert} \\
Sigmund Freud & Johann Wolfgang von Goethe & Isaac Newton & Karl Marx & \textbf{Reese Prosser} \\
J. R. R. Tolkien & Margaret Thatcher & Ptolemy & \textbf{Charles Sanders Peirce} & \textbf{David Tulloch} \\
Johann Wolfgang von Goethe & Ptolemy & Johann Wolfgang von Goethe & Noam Chomsky & \textbf{Kim Hawtrey} \\
\textbf{Spider-Man} & Sigmund Freud & Pliny the Elder & \textbf{Enrico Fermi} & \textbf{Patrick J. Miller} \\
Pliny the Elder & Isaac Newton & Benjamin Franklin & Ptolemy & \textbf{Mikel King} \\
Benjamin Franklin & Benjamin Franklin & J. R. R. Tolkien & \textbf{John Dewey} & \textbf{Albert Perry Brigham} \\
Leonardo da Vinci & J. R. R. Tolkien & Thomas Edison & Johann Wolfgang von Goethe & \textbf{Gordon Wagner (economist)} \\
Isaac Newton & Immanuel Kant & Sigmund Freud & \textbf{Bertrand Russell} & \textbf{George Henry Chase} \\
Ptolemy & Leonardo da Vinci & Immanuel Kant & Plato & \textbf{Charles C. Horn} \\
Immanuel Kant & \textbf{Pierre Andr\'e Latreille} & Leonardo da Vinci &
\textbf{John von Neumann} & \textbf{Paul Goldstene} \\
\textbf{George Bernard Shaw} & Thomas Edison & Noam Chomsky & Vladimir Lenin & \textbf{Robert Stanton Avery} \\
\end{tabular}
\caption{\label{tab:topscientist}Top $20$ pages of Wikipedia following five different centrality measures and restricting pages to Yago2 Wordnet category ``scientist''. The global rank of these items is beyond 300.}
\end{table}

\begin{table}[t]
\centering\small
\begin{tabular}{l|lllll}
&\multicolumn{1}{c}{Ind.} & \multicolumn{1}{c}{PR} & \multicolumn{1}{c}{Katz} & \multicolumn{1}{c}{Harm.} &  \multicolumn{1}{c}{Cl.} \\
\hline
Indegree & 1  & \cellcolor[gray]{0.70} 0.75 & \cellcolor[gray]{0.48} 0.90 & \cellcolor[gray]{0.83} 0.62 & \cellcolor[gray]{0.88} 0.55\\
PageRank & \cellcolor[gray]{0.70} 0.75 & 1  & \cellcolor[gray]{0.70} 0.75 & \cellcolor[gray]{0.84} 0.61 & \cellcolor[gray]{0.87} 0.56\\
Katz & \cellcolor[gray]{0.48} 0.90 & \cellcolor[gray]{0.70} 0.75 & 1  & \cellcolor[gray]{0.76} 0.70 & \cellcolor[gray]{0.84} 0.62\\
Harmonic & \cellcolor[gray]{0.83} 0.62 & \cellcolor[gray]{0.84} 0.61 & \cellcolor[gray]{0.76} 0.70 & 1  & \cellcolor[gray]{0.46} 0.92\\
Closeness & \cellcolor[gray]{0.88} 0.55 & \cellcolor[gray]{0.87} 0.56 & \cellcolor[gray]{0.84} 0.62 & \cellcolor[gray]{0.46} 0.92 & 1 \\
\end{tabular}
\caption{\label{tab:kendall}Kendall's $\tau$ between Wikipedia centrality measures.}
\end{table}

\begin{table}[t]
\centering\small
\begin{tabular}{l|lllll}
&\multicolumn{1}{c}{Ind.} & \multicolumn{1}{c}{PR} & \multicolumn{1}{c}{Katz} & \multicolumn{1}{c}{Harm.} &  \multicolumn{1}{c}{Cl.} \\
\hline
Indegree & 1  & \cellcolor[gray]{0.98} 0.31 & \cellcolor[gray]{0.83} 0.63 & \cellcolor[gray]{0.99} 0.24 & \cellcolor[gray]{1.00} 0.06\\
PageRank & \cellcolor[gray]{0.98} 0.31 & 1  & \cellcolor[gray]{0.99} 0.27 & \cellcolor[gray]{1.00} 0.10 & \cellcolor[gray]{1.00} 0.10\\
Katz & \cellcolor[gray]{0.83} 0.63 & \cellcolor[gray]{0.99} 0.27 & 1  & \cellcolor[gray]{0.91} 0.50 & \cellcolor[gray]{0.99} 0.20\\
Harmonic & \cellcolor[gray]{0.99} 0.24 & \cellcolor[gray]{1.00} 0.10 & \cellcolor[gray]{0.91} 0.50 & 1  & \cellcolor[gray]{0.80} 0.65\\
Closeness & \cellcolor[gray]{1.00} 0.06 & \cellcolor[gray]{1.00} 0.10 & \cellcolor[gray]{0.99} 0.20 & \cellcolor[gray]{0.80} 0.65 & 1 \\
\end{tabular}
\caption{\label{tab:kendall2}Kendall's $\tau$ between Wikipedia centrality
measures, restricted to nodes of indegree 1 and 2.}
\end{table}

\section{Definitions and Tools}
\label{sec:def}

In his 1945 paper about ranking with ties~\cite{KenTTRP}, Kendall, starting from
an observation of Daniels~\cite{DanRMCUSP}, reformulates
his correlation index using a definition similar in spirit to that of an
inner product, which will be the starting point of our proposal: we consider
two real-valued vectors $\br$ and $\bs$ (to be thought as
scores) with indices in $[n]$; then, let us define
\[
\la \br , \bs\ra := \sum_{i<j}\signum(r_i-r_j)\signum(s_i-s_j),
\]
where
\[
\signum(x):=\begin{cases}
1 & \text{if $x > 0$;}\\
0 & \text{if $x = 0$;}\\
-1 & \text{if $x < 0$}.
\end{cases}
\]
Indices of score vectors in summations belong to $[n]$ throughout the paper.
Note that the expression above is actually an inner product in a space of
dimension $n(n-1)$: each score vector $\bm r$ is mapped the vector with
coordinate $\langle i,j\rangle$, $i<j$, given by $\signum(r_i-r_j)$.
We have the property
\[
\la \bm  r , \alpha\bs\ra  = \la  \alpha\br ,
\bs\ra = \signum(\alpha) \la \br , \bs\ra,
\]
which reminds of the analogous property for inner products, and that
$\la\br,-\ra=\la-,\br\ra=0$ if $\br$ is constant.
Following the analogy, we can define
\[
\|\br\| := \sqrt{\la \br, \br\ra},
\]
so 
\[
\|\alpha\br\| = |\signum(\alpha)|\cdot\|\br\|. 
\]
The norm thus defined measures the ``untieness'' of $\br$: it is zero if and only
if $\br$ is a constant vector, and it has maximum value $\sqrt{n(n-1)/2}$ when
all components of $\br$ are distinct.

Since $\la \br, \bs\ra$ is an inner product on a larger space, we have a
Cauchy--Schwartz-like inequality:
\[|\la \br,\bs\ra|\leq \|\br\|\|\bs\|\]
This property makes it possible to define Kendall's $\tau$ between
two vectors $\br$ and $\bs$ with nonnull norm as a normalized inner product, in a way formally identical to
cosine similarity:
\begin{equation}
\label{eq:tau}
\tau(\br, \bs) :=\frac{\la \br, \bs\ra}{\|\br\|\cdot\|\bs\|}.
\end{equation}
We recall that if $\br$ and $\bs$ have no ties, the definition reduces to the
classical ``normalized difference of concordances and discordances'', as the
denominator is exactly $n(n-1)/2$. The definition above is exactly that proposed
by Kendall~\cite{KenTTRP}, albeit we use a different formalism.

The form of~(\ref{eq:tau}) suggests that to obtain a weighted correlation index
it would be natural to define a \emph{weighted} inner product
\[
\la \br , \bs\ra_w := \sum_{i<j}\signum(r_i-r_j)\signum(s_i-s_j)w(i,j),
\]
where $w(-,-):[n]\times [n]\to \mathbf R_{\geq0}$ is some nonnegative
symmetric weight function. We would have then a new norm
$\|\br\|_w=\sqrt{\la\br,\br\ra_w}$ and a new correlation index
\[
\tau_w(\br, \bs) :=\frac{\la \br, \bs\ra_w}{\|\br\|_w\cdot\|\bs\|_w}.
\]
Note that still $\la\br,-\ra_w=\la-,\br\ra_w=0$ if $\br$ is constant.

We say that two score vectors
$\br$ and $\bs$ are \emph{equivalent} if $\signum(r_i-r_j)=\signum(s_i-s_j)$,
\emph{opposite} if $\signum(r_i-r_j)=-\signum(s_i-s_j)$ for all $i$ and $j$.
Since $\la \br, \bs\ra$ is a \emph{semi-definite inner product} on a larger
space (due the possibility of zero weights), the Cauchy-Schwarz inequality
still holds, even if we need positive definiteness for the necessary condition
on equality:
\begin{theorem}
\label{th:cs} 
$|\la \br,\bs\ra_w|\leq \|\br\|_w\|\bs\|_w.$
A sufficient condition for equality to hold is
that the two vectors are equivalent or opposite. 
The condition is necessary if $w$ is strictly positive and $\|\br\|_w,
\|\bs\|_w\neq 0$ or $\|\br\|_w,
\|\bs\|_w=0$
\end{theorem}
Note that the two last conditions are necessary: when $w$ is the constant zero
weight we have equality for all vectors, and if one of the vector has null norm
while the other has not the necessary linearity condition for equality is
moot.

Interestingly, even if our ``inner product'' is neither additive nor linear, we
can still prove directly the triangular inequation for the induced ``norm'':
\begin{theorem}
$\|\br+\bs\|_w\leq \|\br\|_w+\|\bs\|_w$. 
\end{theorem}

\begin{proof}
\begin{align*}
&\|\br+\bs\|_w^2=\la\br + \bs,\br + \bs\ra_w\\
&=\sum_{i<j}\signum(r_i+s_i-r_j-s_j)^2w(i,j)\\
&\leq\sum_{i<j}(|\signum(r_i-r_j)|+|\signum(s_i-s_j)|)^2w(i,j)\\
&= \la\br ,\br\ra_w +\la\bs
,\bs\ra_w+\sum_{i<j}|\signum(r_i-r_j)\signum(s_i-s_j)|w(i,j).
\end{align*}
We now notice that
\[
\sum_{i<j}|\signum(r_i-r_j)\signum(s_i-s_j)|w(i,j)\leq
\sum_{i<j}\signum(r_i-r_j)^2w(i,j)=\la\br ,\br\ra_w=\|\br\|_w^2,
\]
and analogously for $\|\bs\|_w$. We conclude that
\[
\|\br+\bs\|_w^2\leq \|\br\|_w^2 +\|\bs\|_w^2 + 2
\|\br\|_w\|\bs\|_w =(\|\br\|_w +\|\bs\|_w )^2.\qed\]
\end{proof}
The triangular inequality has a nice combinatorial interpretation: adding score vectors can only
\emph{decrease} the amount of ``untieness''. There is no way to induce in a sum vector
more untieness than the amount present in the summands.

Finally, we gather systematically the properties of $\tau_w$:
\begin{theorem}
Let $w:[n]\times [n]\to\R$ be a nonnegative symmetric weight function. The
following properties hold for every score vector $\bm t$ and for every $\br$,
$\bs$ with nonnull norm:
\begin{itemize}
  \item if $\bm t$ is constant, $\|\bm t\|_w=0$;
  \item $-1\leq\tau_w(\br,\bs)\leq1$;
  \item if $\br$ and $\bs$ are equivalent, $\tau_w(\br,\bs)=1$;
  \item if $\br$ and $\bs$ are opposite, $\tau_w(\br,\bs)=-1$;
\end{itemize}

\begin{table*}[t]
\centering
\scriptsize
\begin{tabular}{lllll}
\multicolumn{1}{c}{Indegree} & \multicolumn{1}{c}{PageRank} & \multicolumn{1}{c}{Katz} & \multicolumn{1}{c}{Harmonic} & \multicolumn{1}{c}{Closeness} \\
\hline
Martini (cocktail) & Martini (cocktail) & Irish coffee & Irish coffee & \textbf{Magie Noir} \\
Pi\~na colada & Caipirinha & Caipirinha & Caipirinha & \textbf{Batini (drink)} \\
Mojito & Mojito & Martini (cocktail) & Kir (cocktail) & \textbf{Scorpion bowl} \\
Caipirinha & Pi\~na colada & Pi\~na colada & Martini (cocktail) & \textbf{Poinsettia (cocktail)} \\
Cuba Libre & Irish coffee & Kir (cocktail) & Pi\~na colada & Irish coffee \\
Irish coffee & Kir (cocktail) & Mojito & Mojito & Caipirinha \\
Singapore Sling & Cosmopolitan (cocktail) & Mai Tai & Beer cocktail & Kir (cocktail) \\
Manhattan (cocktail) & Manhattan (cocktail) & Cuba Libre & Shaken, not stirred & Martini (cocktail) \\
Windle (sidecar) & IBA Official Cocktail & Singapore Sling & Pisco Sour & Pi\~na colada \\
Cosmopolitan (cocktail) & Beer cocktail & Long Island Iced Tea & Mai Tai & Mojito \\
Mai Tai & Mai Tai & Shaken, not stirred & Spritz (alcoholic beverage) & Beer cocktail \\
IBA Official Cocktail & Singapore Sling & Beer cocktail & Long Island Iced Tea & Shaken, not stirred \\
Kir (cocktail) & Cuba Libre & Manhattan (cocktail) & Sazerac & Mai Tai \\
Shaken, not stirred & \textbf{Tom Collins} & Cosmopolitan (cocktail) & Fizz (cocktail) & Spritz (alcoholic beverage) \\
Beer cocktail & Long Island Iced Tea & Windle (sidecar) & Flaming beverage & Pisco Sour \\
Pisco Sour & Sour (cocktail) & Pisco Sour & Cuba Libre & Long Island Iced Tea \\
Long Island Iced Tea & Shaken, not stirred & White Russian (cocktail) & Wine cocktail & Sazerac \\
Sour (cocktail) & \textbf{Negroni} & IBA Official Cocktail & Singapore Sling & Flaming beverage \\
White Russian (cocktail) & Flaming beverage & Moscow mule & Moscow mule & Fizz (cocktail) \\
Vesper (cocktail) & \textbf{Lillet} & Vesper (cocktail) & White Russian (cocktail) & Wine cocktail \\
\end{tabular}
\caption{\label{tab:topcocktail}Top $20$ pages of Wikipedia following five different centrality measures and restricting pages to Yago2 Wordnet category ``cocktail''. The global rank of these items is beyond 100\,000.}
\end{table*}

Moreover, if $w$ is strictly positive:
\begin{itemize}
  \item if $\|\bm t\|_w=0$, $\bm t$ is constant;
  \item if $\tau_w(\br,\bs)=1$, $\br$ and $\bs$ are equivalent;
  \item if $\tau_w(\br,\bs)=-1$, $\br$ and $\bs$ are opposite.
\end{itemize}
\end{theorem}
As a result, if $w$ is strictly positive and we obtain correlation $\pm1$
the equivalence classes formed by tied scores are necessarily in a
size-preserving bijection that is monotone decreasing on the scores.

\subsection{Decoupling rank and weight} The reader has
probably already noticed that the dependence on the weight on the \emph{indices} associated to the elements has no meaning: a trivial request
(see, for instance~\cite{KemMWN}) on a correlation measure is that, like
Kendall's $\tau$, it is \emph{invariant by isomorphism}, that is, it does not change if we permute the
indices of the vector. This currently doesn't happen because we are using
the numbering of the element as \emph{ground truth} to weigh the
correlation between $\br$ and $\bs$. While there is nothing bad in principle
(we can stipulate that elements are indexed in order of importance using some external
source of information), we think that a more flexible approach decouples the
problem of the ground truth from the problem of weighting. We thus define the
\emph{ranked-weight} product
\[
\la \br , \bs\ra_{\rho,w} := \sum_{i<j}\signum(r_i-r_j)\signum(s_i-s_j)w(\rho(i),\rho(j)),
\]
where $\rho: [n]\to [n]\cup \{\,\infty\,\}$ is a ranking function
associating with each index a \emph{rank}, 
the highest rank being zero. We admit the possibility of rank $\infty$,
given that the weight function provides a meaningful value in such a case, to
include also the case of \emph{partial ground truths}.
The definition of the ranked-weighted product induces, as in~(\ref{eq:tau}), a correlation index
$\tau_{\rho, w}$, and the machinery we developed applies immediately, as
$w(\rho(-),\rho(-))$ is just a different weight function.

What if there is no ground truth to rely on? Our best bet is to use the
rankings induced by the vectors $\br$ and $\bs$. Let us denote by $\rho_{\br,
\bs}$ the ranking defined by ordering elements lexicographically with respect to
$\br$ and then $\bs$ in case of a tie (in descending order), and analogously for
$\rho_{\bs, \br}$ (if two elements are at a tie in both vectors, their can be placed in any order, 
as their rank does not influence the value of $\tau_{\rho,w}$).
We define
\begin{equation}
\label{eq:symm}
\tau_{w,\bullet}(\br,\bs) := \frac{\tau_{\rho_{\br, \bs}, w}(\br,\bs) +
\tau_{\rho_{\bs, \br}, w}(\br,\bs)}2.
\end{equation}
The same approach has been used in~\cite{YARNRCCIR} to make AP correlation symmetric.
This is the definition used in the rest of the paper.

\subsection{Choosing a weighting scheme}
\label{sec:choosing}

There are of course many ways to choose $w$. For computational reasons, we will
see that it is a good idea to restrict to a class of weighting schemes in which
$w$ is obtained by combining additively or multiplicatively 
a one-argument weighting function $f:[n]\to \R_{\geq0}$ applied to each element
of a pair.

Shieh~\cite{ShiWKTS}, for instance, combines weights multiplicatively, without
giving a motivation. We have, however, two important motivations for
\emph{adding} weights. First and foremost, unless weights are scaled in some way
that depends on $n$ (which we would like to avoid), the largest weight will be
some constant, and then weight will decrease monotonically with importance.
As a result, an exchange between the first and the last element would be assigned an
extremely low weight. Second, adding weights paves the way to a natural
measure for \emph{top $k$ correlation}~\cite{FKSCTKL} by assigning rank $\infty$
to elements after the first $k$. The definition of such a measure in the
multiplicative case is quite contrived and ends up being case-by-case.

For what matters $f$, we are particularly interested in the \emph{hyperbolic}
weight function.
\[
f(r) := \frac1{r+1}.
\]
This function gives more importance to elements of high rank, and weights zero
only pairs in which both index have infinite rank.
Using a hyperbolic weight has a number of useful features. First, it reminds the
well-motivated weight given to exchanges by AP correlation. Second, it
guarantees that as $n$ grows the mass of weight grows indefinitely. Using a
function with quadratic decay, for instance, might end up in making the
influence of low-rank element vanish too quickly, as it is summable. For the opposite
reason, a \emph{logarithmic} decay might fail to be enough discriminative to
provide additional information with respect to the standard $\tau$.

\begin{figure}
\centering
\includegraphics[scale=.3]{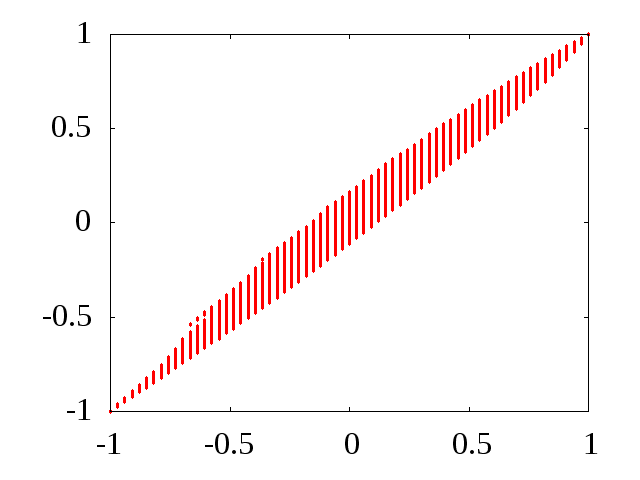}\includegraphics[scale=.3]{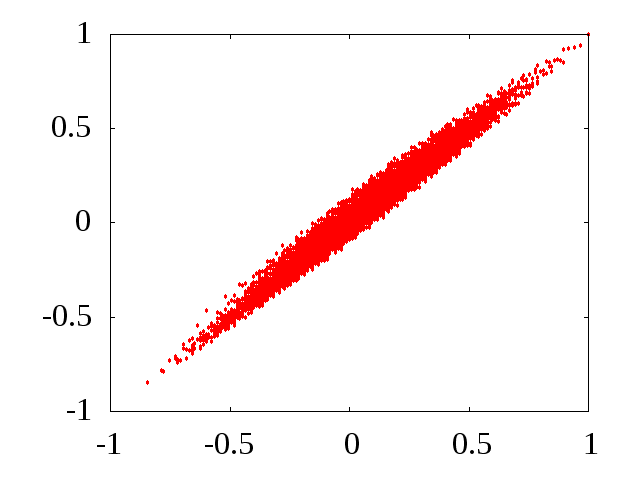}\\
\includegraphics[scale=.3]{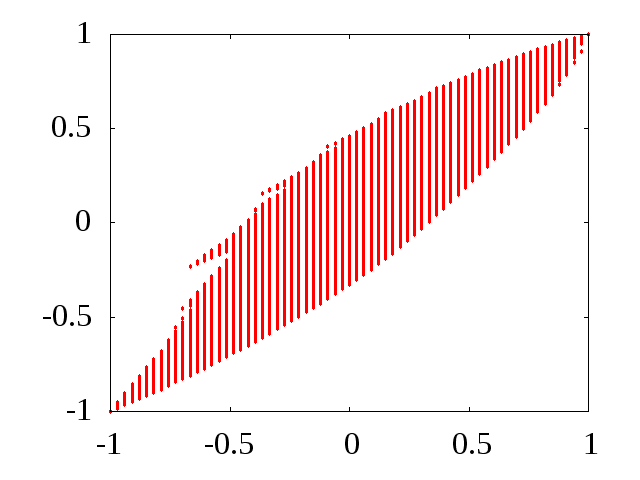}\includegraphics[scale=.3]{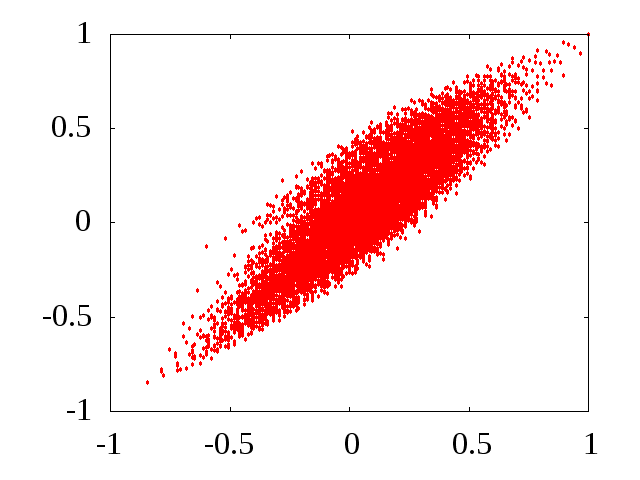}\\
\includegraphics[scale=.3]{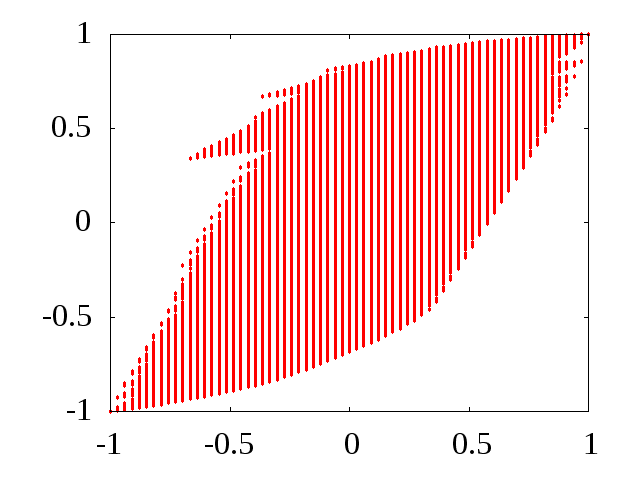}\includegraphics[scale=.3]{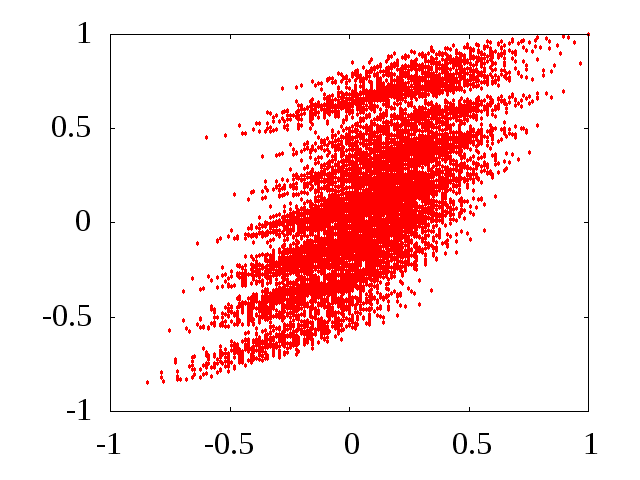}\\
\caption{\label{fig:corr}Scatter plots between Kendall's $\tau$ and the
additive weighted $\tau$.
The rows, from top to bottom, represent logarithmic, hyperbolic and quadratic
weighting.
The plots are generated correlating a permutation of 12 elements
versus the identity permutation (left), 
or a permuted set of scores with skewed distribution w.r.t.~the same scores in descending order (right).}
\end{figure}

We try to make this intuition more concrete in Figure~\ref{fig:corr}, where we
display a number of scatter plots showing the correlation between Kendall's $\tau$ and
the additive weighted $\tau$ defined by (\ref{eq:symm}) under different
weighting schemes.
The left half of the plots correlates all permutations on 12 elements with the
identity permutation. The right half correlates all score vectors made of 15
 values with skewed distribution (there are $t + 1$ elements with score
$0\leq t \leq 4$) with the same vector in descending order. A visual
examination of the plots suggests, indeed, that logarithmic weighting
restricts too much the possible divergence from Kendall's $\tau$, whereas
quadratic weighting ends up in providing answers that are too uncorrelated.
We will return to these consideration in Section~\ref{sec:experiments}.

\section{Computing {\large$\tau_{\rho,\ww}$}}
\label{sec:computing}

Our motivations come from the study of web and social graphs. It is thus essential that
our new correlation measure can be evaluated efficiently. We now
describe a generalization of Knight's algorithm~\cite{KnCMCKTUD} that makes it is possible
to compute $\tau_{\rho,w}$ in time $O(n \log n)$ under some assumptions on $w$.
Our first observation is that, similarly to the unweighted case, each pair of
indices $i,j$ with $i<j$ belongs to one of five subsets; it can be
\begin{itemize}
  \item a \emph{joint tie}, if $r_i=r_j$ and $s_i=s_j$;
  \item a \emph{left tie}, if $r_i=r_j$ and $s_i\neq s_j$;
  \item a \emph{right tie}, if $r_i\neq r_j$ and $s_i = s_j$;
  \item a \emph{concordance}, if $\signum(r_i -r_j)\signum(s_i -s_j)=1$;
  \item a \emph{discordance}, if $\signum(r_i -r_j)\signum(s_i -s_j)=-1$.
\end{itemize}

Let $J$, $L$, $R$, $C$ and $D$ be the overall weight of joint ties, left ties, right ties, concordances and 
discordances, respectively. Clearly,
\[
J+L+R+C+D = \sum_{i<j}w(\rho(i),\rho(j)) =T.
\]
The first requirement for our technique to work is that $T$ can be computed
easily.
This is possible if weights are computed additively or multiplicatively 
from some single-argument function $f$. In the 
additive case,
\begin{equation}
\label{eq:totsum}
T=\sum_{i<j}\bigl(f(\rho(i))+f(\rho(j))\bigr) =
(n-1)\sum_i f(\rho(i)).
\end{equation}
Also the multiplicative case is easy, as
\begin{equation}
\label{eq:totprod}
2T=2\sum_{i< j}f(\rho(i))f(\rho(j)) 
= \biggl(\sum_i f(\rho(i))\biggr)^2-\sum_i f(\rho(i))^2.
\end{equation}
The same observation leads to a simple $O(n\log n)$ algorithm to compute $L$: sort 
the indices in $[n]$ by $\br$, and for
each block of consecutive $k>1$ elements with the same score
apply~(\ref{eq:totsum}) or~(\ref{eq:totprod}) restricting the indices to the
subset. In the same way one can compute $R$ and $J$.

We now observe that, as in the unweighted case,
\[
\la \br , \bs\ra_{\rho, w}  = C - D = T - ( L + R - J ) - 2D.
\]
This can be easily seen from the fact that $C$ is given by the total weight $T$, minus the weight of discordances $D$,
minus the number of ties, joint or not, which is $L+R-J$ (we must avoid to count twice the weight of joint ties, hence the $-J$ term).
In particular,
\[
\la \br , \br\ra_{\rho, w}  = T - L \qquad\la \bs , \bs\ra_{\rho, w}   = T - R,
\]
as in this case there are just concordances and all ties are joint.

We are left with the computation of $D$. The core of Knight's algorithm is an \emph{exchange counter}: an $O(n\log n)$
algorithm that given a list of elements and an order $\preceq$ on the elements
of the list computes the number of exchanges that are necessary to
$\preceq$-sort the list. The algorithm is a modified MergeSort~\cite{KnuACPSS}\footnote{In
principle, any stable algorithm that sorts by comparison could be used. This is particularly interesting as
entirely on-disk algorithms, such as \emph{polyphase merge}~\cite{KnuACPSS},
could be used to count exchanges using constant core memory.}: during
the merging phase, whenever an element is moved from the second list to the
temporary result list the current number of elements of the first list is added to the number
of exchanges. The number of discordances is then equal to the number
of exchanges (as we evaluate whether there is a discordance on $i$ and $j$ only
for $i<j$).

Our goal is to make this computation weighted: for this to happen,
it must be possible to keep track incrementally of a \emph{residual weight} $r$
associated with the first list, and obtain in constant time the weight of the
exchanges generated by the movement of an element from the second list.

If weights are computed multiplicatively or additively starting from a
single-argument function $f$ this is not difficult: it is sufficient to let $r$ be the sum of the
values of $f$ applied to the elements currently in the first list. In the
additive case, moving an element $i$ from the second list 
increases the weight of exchanges by the residual $r$
plus the weight $f(\rho(i))$ multiplied by the length of the first list.
In the multiplicative case, we must instead use the weight $f(\rho(i))$ multiplied by the
residual $r$.  When we
move an element from the first list we update the residual by subtracting its weight.

The resulting recursive procedure (for the additive case) is
Algorithm~\ref{algo:wkt}. The final layout of the computation of
$\tau_{\rho, w}$ is thus as follows:
\begin{itemize}
  \item Consider a list $\mathscr L$ initially filled with the integers $[0\..n)$. 
  \item Sort stably $\mathscr L$ using $\br$ as primary key and $\bs$ as secondary key.  
  \item Compute $T$ and $L$ using $\mathscr L$ to enumerate elements in the
  order defined by $\br$ and $\bs$.
  \item Apply Algorithm~\ref{algo:wkt} to $\mathscr L$ using $\bs$ to define the
  order $\preceq$, thus computing $D$ and sorting $\mathscr L$ by $\bs$.
  \item Compute $R$ using $\mathscr L$ to enumerate elements in the
  order defined by $\bs$.
  \item Compute $T$ and put everything together.
\end{itemize}

\begin{Algorithm}[htb]
\smallskip\textbf{Input:} A list $\mathscr L$, a comparison function $\preceq$ for the elements of $\mathscr L$, 
a rank function $\rho$, and a single-argument weight function $f$ that will be combined additively. $e$ is a global variable
initialized to $0$ that will contain the weight of exchanges after the call weigh($0$, $|\mathscr L|$).
The procedure works on a sublist specified by its starting index $0\leq s<|\mathscr L|$
and its length $\ell$. $\mathscr T$ is a temporary list.\\
\textbf{Output:} the sum of $f(\rho(-))$ on the specified sublist.\\ 
\vspace{-1.5em}
\begin{tabbing}
\setcounter{prgline}{0}
\hspace{0.5cm} \= \hspace{0.3cm} \= \hspace{0.3cm} \= \hspace{0.3cm} \= \hspace{0.3cm} \= \kill
\pl\>\FUNCTION weigh($s$ : integer, $\ell$ : integer)\\
\pl\>\>\label{pr:g1}  \IF $\ell = 1$ \THEN \RETURN $f(\rho(\mathscr L[s]))$
\FI\\
\pl\>\>  $\ell_0 \leftarrow \lfloor\ell/2\rfloor$		 \\
\pl\>\>  $\ell_1 \leftarrow \ell-\ell_0$ \\
\pl\>\>  $m \leftarrow s + \ell_0$ \\
\pl\>\>  $r \leftarrow \operatorname{weigh}(s,\ell_0)$\\	
\pl\>\>  $w \leftarrow \operatorname{weigh}(m,\ell_1) + r$\\	
\pl\>\>  $i, j, k\leftarrow 0$			\\
\pl\>\>  \WHILE $j < \ell_0$ \LAND $k < \ell_1$ \DO\\
\pl\>\> 	\> \IF $\mathscr L[s+j] \preceq \mathscr L[m+k]$ \THEN\\
\pl\>\>  \> \>$\mathscr T[i] = \mathscr L[ s + j\pp ]$\\
\pl\>\>\label{pr:g2}   \> \>$r \leftarrow r - f(\rho(\mathscr T[i]))$\\
\pl\>\>  \> \ELSE\\
\pl\>\>  \> \>$\mathscr T[i] = \mathscr L[ m + k\pp ]$\\
\pl\>\>  \> \>$e \leftarrow e + f(\rho(\mathscr T[i]))\cdot (\ell_0 - j ) + r$\\
\pl\>\>  \> \FI\\
\pl\>\>  \> $i\pp$\\
\pl\>\>  \OD\\
\pl\>\>  \FOR $k = \ell_0 - j - 1,\ldots, 0$ \DO\\
\pl\>\>\>  $\mathscr L[ s + i + k ]\leftarrow \mathscr L[ s + j + k ]$\\ 
\pl\>\>  \OD \\
\pl\>\>  \FOR $k = 0, \ldots, i- 1$ \DO $\mathscr L[ s + k ] \leftarrow
\mathscr T[k]$ \OD \\
\pl\>\>  \RETURN $w$\\
\pl\>\END
\end{tabbing}
\caption{\label{algo:wkt}A generalization of Knight's algorithm for weighing
exchanges.}
\end{Algorithm}

\begin{Algorithm}[htb]
\begin{tabbing}
\setcounter{prgline}{9}
\hspace{0.5cm} \= \hspace{0.3cm} \= \hspace{0.3cm} \= \hspace{0.3cm} \= \hspace{0.3cm} \= \kill
\pl\>  \IF $\mathscr L[s+j] \preceq \mathscr L[m+k]$ \THEN\\
\pl\>  \>$\mathscr T[i] = \mathscr L[ s + j\pp ]$\\
\pl\>  \ELSE\\
\pl\>  \>$\mathscr T[i] = \mathscr L[ m + k\pp ]$\\
\pl\>  \>$e \leftarrow e + (\ell_0 - j)/\rho(\mathscr T[i])$\\
\pl\>  \FI
\end{tabbing}
\caption{\label{algo:ap}The replacement for lines 9--15 of
Algorithm~\ref{algo:wkt} to compute AP correlation.}
\end{Algorithm}

The running time of the computation is dominated by the sorting phases, and
it is thus $O(n\log n)$.

\subsection{The asymmetric case and AP Correlation}

It is easy to adapt Algorithm~\ref{algo:wkt} for the case in which $w(i,j)$ is
given by a combination of \emph{two} different
one-argument functions, one, $f$, for the left index and one, $g$, for the right
index. The only modification of Algorithm~\ref{algo:wkt} is the replacement
of $f$ with $g$ at line~14, so that we combine the residual computed with
$f$ with a weight computed with $g$.

The formulae for computing $T$ can be updated easily for the additive case: 
\[
T=\sum_{i<j}\bigl(f(\rho(i))+g(\rho(j))\bigr) = \sum_{i\neq0} i(f(\rho(n-1-i))
+ g(\rho(i)))
\]
and for the multiplicative case:
\[
T=\sum_{i<j}f(\rho(i))g(\rho(j)) = \sum_i f(\rho(i))
\sum_{i<j}g(\rho(j)).
\]
Both formulae can be computed in linear time using a suitable loop.

Given this setup, it is easy to compute AP correlation:
as it can be checked from the very definition~\cite{YARNRCCIR}, the AP
correlation of $\br$ w.r.t.~$\bs$, where both vectors have no ties, is simply
$\tau_{w,\rho_{\bs}}(\br,\bs)$, where $\rho_{\bs}$ is the ranking
induced by $\bs$ and the weight function $w$ is computed additively from
two weight functions $f(r)=0$, $g(r)=1/r$. In this
case, $T=n-1$, $J=L=R=0$ (we are under the assumption that there are no ties)
and Algorithm~\ref{algo:wkt} can be considerably simplified, as the residual $r$
is always zero.\footnote{Of course, it is possible to forget that we are
computing AP correlation and use the weight matrix just described combined with
the machinery of Section~\ref{sec:def} to define an ``AP correlation with
ties''. In this case, $J$, $L$ and $R$ should be computed using the formulae
for the asymmetric case, and the probabilistic interpretation would be lost.
Such an index would probably give a notion of correlation very similar to $\tauh$, but we find more natural and more in line
with Kendall's original definition to introduce the weighted $\tau$ as a
symmetric index in which both ends of an exchange are relevant in computing the
exchange weight.}

Algorithm~\ref{algo:ap} makes explicit the change to the selection statement of
Algorithm~\ref{algo:wkt} that is sufficient to compute AP correlation. Since
keeping track of the residual is no longer necessary, the recursive function
can be further simplified to a recursive procedure that does not return a
value.
The value $e$ computed by the modified algorithm is all we need to compute AP correlation using the formula
$(T-2e)/T$.

\begin{table}[t]
\centering\small
\begin{tabular}{l|lllll}
&\multicolumn{1}{c}{Ind.} & \multicolumn{1}{c}{PR} & \multicolumn{1}{c}{Katz} & \multicolumn{1}{c}{Harm.} &  \multicolumn{1}{c}{Cl.} \\
\hline
Indegree & 1  & \cellcolor[gray]{0.40} 0.95 & \cellcolor[gray]{0.35} 0.98 & \cellcolor[gray]{0.48} 0.90 & \cellcolor[gray]{0.99} 0.27\\
PageRank & \cellcolor[gray]{0.40} 0.95 & 1  & \cellcolor[gray]{0.39} 0.96 & \cellcolor[gray]{0.46} 0.92 & \cellcolor[gray]{0.81} 0.65\\
Katz & \cellcolor[gray]{0.35} 0.98 & \cellcolor[gray]{0.39} 0.96 & 1  & \cellcolor[gray]{0.43} 0.93 & \cellcolor[gray]{0.99} 0.26\\
Harmonic & \cellcolor[gray]{0.48} 0.90 & \cellcolor[gray]{0.46} 0.92 & \cellcolor[gray]{0.43} 0.93 & 1  & \cellcolor[gray]{0.98} 0.28\\
Closeness & \cellcolor[gray]{0.99} 0.27 & \cellcolor[gray]{0.81} 0.65 & \cellcolor[gray]{0.99} 0.26 & \cellcolor[gray]{0.98} 0.28 & 1 \\
\end{tabular}
\caption{\label{tab:tauh}$\tauh$ on Wikipedia.}
\end{table}

\section{Experiments}
\label{sec:experiments}

We now return to our main motivation---understanding the correlation between
centralities on large graph. In this section, we gather the results of a number of computational
experiment that help to corroborate our intuition that $\tauh$, the
\emph{additive hyperbolic weighted} $\tau$, works as expected. We will find also
an interesting surprise along the way.

Note that judging whether a new measure is useful for such a purpose is a
difficult task: to be interesting, a new measure must highlight features that
were previously undetectable or badly evaluated, but those are exactly those
features on which a systematic assessment is problematic.

Table~\ref{tab:tauh} reports the value of $\tauh$
on the Wikipedia graph. We finally see data corresponding to the empirical
evidence discussed in Section~\ref{sec:motivation}: indegree, Katz and PageRank
are almost identical, harmonic centrality is highly correlated but definitely
less than the previous triple, which matches our empirical observations. Closeness is not close to any ranking (and in
particular, not to harmonic centrality) due to its pathological
behavior.

There is of course a value that immediately stands out: the suspiciously
high correlation ($0.65$) between closeness and PageRank. We reserve discussing this
value for later.

In Table~\ref{tab:taulq} we show the same data for 
logarithmic and quadratic weights. The intuition we gathered from
Figure~\ref{fig:corr} is fully confirmed: logarithmic weights provides
results almost indistinguishable from Kendall's $\tau$
 (compare with Table~\ref{tab:kendall}), and quadratic weighs
make the influence of the tail so low that all non-pathological scores collapse.
 
To gather a better understanding of the behavior of $\tauh$ we extended our
experiments to two very different datasets:
the \emph{Hollywood co-starship graph}, an undirected graph ($2$ million nodes,
$229$ million edges) with an edge between two persons appearing in the
Internet Movie Data Base if they ever worked together, and a \emph{host
graph} ($100$ million nodes, $2$ billion arcs) obtained from a large-scale crawl
gathered by the Common Crawl Foundation\footnote{\url{http://commoncrawl.org/}} in the first half of
2012.\footnote{The crawl contains $3.53$ billion web documents; we are
using the associated host graph, which has a node for each host and an arc
between two hosts $x$ and $y$ if some page in $x$ points to some page in
$y$. More information about the graph can be found in \cite{MVLGSWR}, and the
complete host ranking can be accessed at
\url{http://wwwranking.webdatacommons.org/}.} As (unavoidably anecdotal)
empirical evidence we report the top 20 nodes for both graphs.

Table~\ref{tab:tauhollywood} should be compared with Table~\ref{tab:tophollywood}.
PageRank and harmonic centrality turns to be less correlated to indegree than
Katz in Table~\ref{tab:tauhollywood}, and indeed we find many quirk choices in the very
top PageRank actors (Ron Jeremy is a famous porn star; Lloyd Kaufman is an
independent horror/splatter filmmaker and Debbie Rochon an actress working with him).
Harmonic centrality provides unique names such as Malcolm McDowell, Robert De
Niro, Anthony Hopkins and Sylvester Stallone, and drops all USA presidents altogether.
Kendall's $\tau$ values, instead, suggest that PageRank and harmonic
centrality are entirely uncorrelated (whereas we find several common items),
and that harmonic and closeness centrality should be extremely similar.

We see analogous results comparing Table~\ref{tab:tauhcc} with
Table~\ref{tab:topcc}. Here $\tauh$ separates in a very strong way harmonic
centrality from the first three, and indeed we see a significant
difference in the lists, with numerous sites that have a high indegree and appear in at
least two of the three lists because of technical or political reasons
(\texttt{\small gmpg.org}, \texttt{\small rtalabel.org}, \texttt{\small staff.tumblr.com}, \texttt{\small miibeian.gov.cn}, \texttt{\small phpbb.com}) 
disappearing altogether in favor of sites such as \texttt{\small apple.com},
\texttt{\small amazon.com}, \texttt{\small myspace.com}, \texttt{\small microsoft.com},
\texttt{\small bbc.co.uk}, \texttt{\small nytimes.com} and \texttt{\small guardian.co.uk}, which
do not appear in any other list. If we look at Kendall's $\tau$, we should
expect PageRank and Katz to give very different rankings, whereas more than half
of their top 20 elements are in common.

\begin{table}[t]
\centering\small
\begin{tabular}{l|lllll}
&\multicolumn{1}{c}{Ind.} & \multicolumn{1}{c}{PR} & \multicolumn{1}{c}{Katz} & \multicolumn{1}{c}{Harm.} &  \multicolumn{1}{c}{Cl.} \\
\hline
Indegree & 1  & \cellcolor[gray]{0.70} 0.76 & \cellcolor[gray]{0.48} 0.90 & \cellcolor[gray]{0.83} 0.63 & \cellcolor[gray]{0.88} 0.55\\
PageRank & \cellcolor[gray]{0.70} 0.76 & 1  & \cellcolor[gray]{0.69} 0.76 & \cellcolor[gray]{0.84} 0.62 & \cellcolor[gray]{0.87} 0.56\\
Katz & \cellcolor[gray]{0.48} 0.90 & \cellcolor[gray]{0.69} 0.76 & 1  & \cellcolor[gray]{0.76} 0.70 & \cellcolor[gray]{0.84} 0.62\\
Harmonic & \cellcolor[gray]{0.83} 0.63 & \cellcolor[gray]{0.84} 0.62 & \cellcolor[gray]{0.76} 0.70 & 1  & \cellcolor[gray]{0.47} 0.91\\
Closeness & \cellcolor[gray]{0.88} 0.55 & \cellcolor[gray]{0.87} 0.56 & \cellcolor[gray]{0.84} 0.62 & \cellcolor[gray]{0.47} 0.91 & 1 \\
\end{tabular}\\[2ex]
\begin{tabular}{l|lllll}
&\multicolumn{1}{c}{Ind.} & \multicolumn{1}{c}{PR} & \multicolumn{1}{c}{Katz} & \multicolumn{1}{c}{Harm.} &  \multicolumn{1}{c}{Cl.} \\
\hline
Indegree & 1  & \cellcolor[gray]{0.30} 1.00 & \cellcolor[gray]{0.30} 1.00 & \cellcolor[gray]{0.30} 1.00 & \cellcolor[gray]{0.99} 0.22\\
PageRank & \cellcolor[gray]{0.30} 1.00 & 1  & \cellcolor[gray]{0.30} 1.00 & \cellcolor[gray]{0.30} 1.00 & \cellcolor[gray]{0.57} 0.85\\
Katz & \cellcolor[gray]{0.30} 1.00 & \cellcolor[gray]{0.30} 1.00 & 1  & \cellcolor[gray]{0.30} 1.00 & \cellcolor[gray]{1.00} 0.18\\
Harmonic & \cellcolor[gray]{0.30} 1.00 & \cellcolor[gray]{0.30} 1.00 & \cellcolor[gray]{0.30} 1.00 & 1  & \cellcolor[gray]{1.00} 0.07\\
Closeness & \cellcolor[gray]{0.99} 0.22 & \cellcolor[gray]{0.57} 0.85 & \cellcolor[gray]{1.00} 0.18 & \cellcolor[gray]{1.00} 0.07 & 1 \\
\end{tabular}
\caption{\label{tab:taulq}The logarithmic (top) and quadratic (bottom) additive
$\tau$ on Wikipedia.}
\end{table}

\begin{table}[t]
\centering\small
\begin{tabular}{l|lllll}
&\multicolumn{1}{c}{Ind.} & \multicolumn{1}{c}{PR} & \multicolumn{1}{c}{Katz} & \multicolumn{1}{c}{Harm.} &  \multicolumn{1}{c}{Cl.} \\
\hline
Indegree & 1  & \cellcolor[gray]{0.95} 0.42 & \cellcolor[gray]{0.44} 0.93 & \cellcolor[gray]{0.88} 0.55 & \cellcolor[gray]{0.94} 0.43\\
PageRank & \cellcolor[gray]{0.95} 0.42 & 1  & \cellcolor[gray]{0.97} 0.36 & \cellcolor[gray]{1.00} 0.10 & \cellcolor[gray]{1.00} 0.18\\
Katz & \cellcolor[gray]{0.44} 0.93 & \cellcolor[gray]{0.97} 0.36 & 1  & \cellcolor[gray]{0.84} 0.61 & \cellcolor[gray]{0.92} 0.49\\
Harmonic & \cellcolor[gray]{0.88} 0.55 & \cellcolor[gray]{1.00} 0.10 & \cellcolor[gray]{0.84} 0.61 & 1  & \cellcolor[gray]{0.56} 0.86\\
Closeness & \cellcolor[gray]{0.94} 0.43 & \cellcolor[gray]{1.00} 0.18 & \cellcolor[gray]{0.92} 0.49 & \cellcolor[gray]{0.56} 0.86 & 1 \\
\end{tabular}\\[2ex]
\begin{tabular}{l|lllll}
&\multicolumn{1}{c}{Ind.} & \multicolumn{1}{c}{PR} & \multicolumn{1}{c}{Katz} & \multicolumn{1}{c}{Harm.} &  \multicolumn{1}{c}{Cl.} \\
\hline
Indegree & 1  & \cellcolor[gray]{0.49} 0.90 & \cellcolor[gray]{0.33} 0.98 & \cellcolor[gray]{0.48} 0.91 & \cellcolor[gray]{1.00} 0.10\\
PageRank & \cellcolor[gray]{0.49} 0.90 & 1  & \cellcolor[gray]{0.52} 0.88 & \cellcolor[gray]{0.62} 0.81 & \cellcolor[gray]{0.81} 0.64\\
Katz & \cellcolor[gray]{0.33} 0.98 & \cellcolor[gray]{0.52} 0.88 & 1  & \cellcolor[gray]{0.45} 0.92 & \cellcolor[gray]{1.00} 0.11\\
Harmonic & \cellcolor[gray]{0.48} 0.91 & \cellcolor[gray]{0.62} 0.81 & \cellcolor[gray]{0.45} 0.92 & 1  & \cellcolor[gray]{1.00} 0.18\\
Closeness & \cellcolor[gray]{1.00} 0.10 & \cellcolor[gray]{0.81} 0.64 & \cellcolor[gray]{1.00} 0.11 & \cellcolor[gray]{1.00} 0.18 & 1 \\
\end{tabular}
\caption{\label{tab:tauhollywood}Kendall's $\tau$ (top) and $\tauh$ (bottom) on
the Hollywood co-starship graph.}
\end{table}

\begin{table}[t]
\centering\small
\begin{tabular}{l|lllll}
&\multicolumn{1}{c}{Ind.} & \multicolumn{1}{c}{PR} & \multicolumn{1}{c}{Katz} & \multicolumn{1}{c}{Harm.} &  \multicolumn{1}{c}{Cl.} \\
\hline
Indegree & 1  & \cellcolor[gray]{0.75} 0.71 & \cellcolor[gray]{0.51} 0.89 & \cellcolor[gray]{0.84} 0.61 & \cellcolor[gray]{0.89} 0.54\\
PageRank & \cellcolor[gray]{0.75} 0.71 & 1  & \cellcolor[gray]{0.80} 0.66 & \cellcolor[gray]{0.91} 0.50 & \cellcolor[gray]{0.91} 0.50\\
Katz & \cellcolor[gray]{0.51} 0.89 & \cellcolor[gray]{0.80} 0.66 & 1  & \cellcolor[gray]{0.77} 0.69 & \cellcolor[gray]{0.86} 0.59\\
Harmonic & \cellcolor[gray]{0.84} 0.61 & \cellcolor[gray]{0.91} 0.50 & \cellcolor[gray]{0.77} 0.69 & 1  & \cellcolor[gray]{0.55} 0.86\\
Closeness & \cellcolor[gray]{0.89} 0.54 & \cellcolor[gray]{0.91} 0.50 & \cellcolor[gray]{0.86} 0.59 & \cellcolor[gray]{0.55} 0.86 & 1 \\
\end{tabular}\\[2ex]
\begin{tabular}{l|lllll}
&\multicolumn{1}{c}{Ind.} & \multicolumn{1}{c}{PR} & \multicolumn{1}{c}{Katz} & \multicolumn{1}{c}{Harm.} &  \multicolumn{1}{c}{Cl.} \\
\hline
Indegree & 1  & \cellcolor[gray]{0.48} 0.91 & \cellcolor[gray]{0.37} 0.96 & \cellcolor[gray]{0.74} 0.72 & \cellcolor[gray]{0.99} 0.20\\
PageRank & \cellcolor[gray]{0.48} 0.91 & 1  & \cellcolor[gray]{0.50} 0.90 & \cellcolor[gray]{0.62} 0.81 & \cellcolor[gray]{0.77} 0.69\\
Katz & \cellcolor[gray]{0.37} 0.96 & \cellcolor[gray]{0.50} 0.90 & 1  & \cellcolor[gray]{0.67} 0.78 & \cellcolor[gray]{1.00} 0.15\\
Harmonic & \cellcolor[gray]{0.74} 0.72 & \cellcolor[gray]{0.62} 0.81 & \cellcolor[gray]{0.67} 0.78 & 1  & \cellcolor[gray]{0.97} 0.35\\
Closeness & \cellcolor[gray]{0.99} 0.20 & \cellcolor[gray]{0.77} 0.69 & \cellcolor[gray]{1.00} 0.15 & \cellcolor[gray]{0.97} 0.35 & 1 \\
\end{tabular}
\caption{\label{tab:tauhcc}Kendall's $\tau$ (top) and $\tauh$ (bottom) on the on
the Common Crawl host graph.}
\end{table}

\begin{table*}[htb]
\centering
\scriptsize
\begin{tabular}{lllll}
\multicolumn{1}{c}{Indegree} & \multicolumn{1}{c}{PageRank} & \multicolumn{1}{c}{Katz} & \multicolumn{1}{c}{Harmonic} & \multicolumn{1}{c}{Closeness} \\
\hline
Shatner, William & Jeremy, Ron & Shatner, William & Sheen, Martin & \textbf{\"Ostlund, Claes G\"oran} \\
Flowers, Bess & Hitler, Adolf & Sheen, Martin & Clooney, George & \textbf{\"Ostlund, Catarina} \\
Sheen, Martin & \textbf{Kaufman, Lloyd} & Hanks, Tom & Jackson, Samuel L. & \textbf{von Preu\ss en, Oskar Prinz} \\
Reagan, Ronald (I) & Bush, George W. & Williams, Robin (I) & Hopper, Dennis & \textbf{von Preu\ss en, Georg Friedrich} \\
Clooney, George & Reagan, Ronald (I) & Clooney, George & Hanks, Tom & \textbf{von Mannstein, Robert Grund} \\
Jackson, Samuel L. & Clinton, Bill (I) & Reagan, Ronald (I) & Stone, Sharon (I) & \textbf{von Mannstein, Concha} \\
Williams, Robin (I) & Sheen, Martin & Willis, Bruce & Brosnan, Pierce & \textbf{von der Busken, Mart} \\
Hanks, Tom & \textbf{Rochon, Debbie} & Jackson, Samuel L. & Hitler, Adolf & \textbf{van der Putten, Thea} \\
Jeremy, Ron & \textbf{Kennedy, John F.} & Stone, Sharon (I) & \textbf{McDowell, Malcolm} & \textbf{de la Bruheze, Joel Albert} \\
Hitler, Adolf & Hopper, Dennis & Freeman, Morgan (I) & Williams, Robin (I) & \textbf{de la Bruheze, Emile} \\
Willis, Bruce & \textbf{Nixon, Richard} & Flowers, Bess & \textbf{De Niro, Robert} & \textbf{te Riele, Marloes} \\
Clinton, Bill (I) & \textbf{Estevez, Joe} & Brosnan, Pierce & Willis, Bruce & \textbf{de Reijer, Eric} \\
Freeman, Morgan (I) & Shatner, William & Douglas, Michael (I) & \textbf{Hopkins, Anthony} & \textbf{des Bouvrie, Jan} \\
Hopper, Dennis & Jackson, Samuel L. & Madonna (I) & Madonna (I) & \textbf{de Klijn, Judith} \\
Stone, Sharon (I) & \textbf{Stewart, Jon (I)} & Travolta, John & \textbf{Lee, Christopher (I)} & \textbf{de Freitas, Lu\'is (II)} \\
Madonna (I) & \textbf{Carradine, David (I)} & Hopper, Dennis & Douglas, Michael (I) & \textbf{de Freitas, Lu\'is (I)} \\
Bush, George W. & Asner, Edward & Ford, Harrison (I) & \textbf{Sutherland, Donald (I)} & \textbf{Zuu, Winnie Otondi} \\
\textbf{Harris, Sam (II)} & \textbf{Zirnkilton, Steven} & Asner, Edward & Freeman, Morgan (I) & \textbf{Zuu, Emmanuel Dahngbay} \\
Brosnan, Pierce & \textbf{Colbert, Stephen} & \textbf{MacLaine, Shirley} & \textbf{Stallone, Sylvester} & \textbf{Zilbersmith, Carla} \\
Travolta, John & \textbf{Madsen, Michael (I)} & Clinton, Bill (I) & Ford, Harrison (I) & \textbf{Zilber, Mac} \\
\end{tabular}
\caption{\label{tab:tophollywood}Top 20 pages of the Hollywood co-starship
graph.}
\end{table*}

\begin{table*}[htb]
\centering
\scriptsize
\begin{tabular}{lllll}
\multicolumn{1}{c}{Indegree} & \multicolumn{1}{c}{PageRank} & \multicolumn{1}{c}{Katz} & \multicolumn{1}{c}{Harmonic} & \multicolumn{1}{c}{Closeness} \\
\hline
\texttt{wordpress.org} & \texttt{gmpg.org} & \texttt{wordpress.org} & \texttt{youtube.com} & \textbf{\texttt{0--p.com}} \\
\texttt{youtube.com} & \texttt{wordpress.org} & \texttt{youtube.com} & \texttt{en.wikipedia.org} & \textbf{\texttt{0-0-0-0-0-0-0.indahiphop.ru}} \\
\texttt{gmpg.org} & \texttt{youtube.com} & \texttt{gmpg.org} & \texttt{twitter.com} & \textbf{\texttt{0-0-1.i.tiexue.net}} \\
\texttt{en.wikipedia.org} & \textbf{\texttt{livejournal.com}} & \texttt{en.wikipedia.org} & \texttt{google.com} & \textbf{\texttt{0-00cigarettes.info}} \\
\texttt{tumblr.com} & \texttt{tumblr.com} & \texttt{tumblr.com} & \texttt{wordpress.org} & \textbf{\texttt{0-0mos00.hi5.com}} \\
\texttt{twitter.com} & \texttt{en.wikipedia.org} & \texttt{twitter.com} & \texttt{flickr.com} & \textbf{\texttt{0-0new0-0.hi5.com}} \\
\texttt{google.com} & \texttt{twitter.com} & \texttt{google.com} & \texttt{facebook.com} & \textbf{\texttt{0-0sunny0-0.hi5.com}} \\
\texttt{flickr.com} & \textbf{\texttt{networkadvertising.org}} & \texttt{flickr.com} & \textbf{\texttt{apple.com}} & \textbf{\texttt{0-1.i.tiexue.net}} \\
\texttt{rtalabel.org} & \textbf{\texttt{promodj.com}} & \texttt{rtalabel.org} & \texttt{vimeo.com} & \textbf{\texttt{0-1.sxsy.co}} \\
\texttt{wordpress.com} & \textbf{\texttt{skriptmail.de}} & \texttt{wordpress.com} & \texttt{creativecommons.org} & \textbf{\texttt{0-2.paparazziwannabe.com}} \\
\texttt{mp3shake.com} & \textbf{\texttt{parallels.com}} & \texttt{mp3shake.com} & \textbf{\texttt{amazon.com}} & \textbf{\texttt{0-311.cn}} \\
\texttt{w3schools.com} & \textbf{\texttt{tistory.com}} & \texttt{w3schools.com} & \textbf{\texttt{adobe.com}} & \textbf{\texttt{0-360.rukazan.ru}} \\
\texttt{domains.lycos.com} & \texttt{google.com} & \texttt{creativecommons.org} & \textbf{\texttt{myspace.com}} & \textbf{\texttt{0-5days.com}} \\
\texttt{staff.tumblr.com} & \texttt{miibeian.gov.cn} & \texttt{staff.tumblr.com} & \textbf{\texttt{w3.org}} & \textbf{\texttt{0-5days.net}} \\
\texttt{club.tripod.com} & \texttt{phpbb.com} & \texttt{domains.lycos.com} & \textbf{\texttt{bbc.co.uk}} & \textbf{\texttt{0-5kalibr.pdj.ru}} \\
\texttt{creativecommons.org} & \textbf{\texttt{blog.fc2.com}} & \texttt{club.tripod.com} & \textbf{\texttt{nytimes.com}} & \textbf{\texttt{0-9-0-4-4-9.promoradio.ru}} \\
\texttt{vimeo.com} & \textbf{\texttt{tw.yahoo.com}} & \texttt{vimeo.com} & \textbf{\texttt{yahoo.com}} & \textbf{\texttt{0-9-0-9.dbass.ru}} \\
\texttt{miibeian.gov.cn} & \texttt{w3schools.com} & \texttt{miibeian.gov.cn} & \textbf{\texttt{microsoft.com}} & \textbf{\texttt{0-9-0-9.promodj.ru}} \\
\texttt{facebook.com} & \texttt{wordpress.com} & \texttt{facebook.com} & \textbf{\texttt{guardian.co.uk}} & \textbf{\texttt{0-9-1125.i.tiexue.net}} \\
\texttt{phpbb.com} & \texttt{domains.lycos.com} & \texttt{phpbb.com} & \textbf{\texttt{imdb.com}} & \textbf{\texttt{0-9-7-16.software.informer.com}} \\
\end{tabular}
\caption{\label{tab:topcc}Top 20 hosts of the Common Crawl host graph.}
\end{table*}

\subsection{PageRank and closeness}

It is now time to examine the mysteriously high $\tauh$ between PageRank and
closeness we found in all our graphs. When we first computed our correlation
tables, we were puzzled by its value. The phenomenon is interesting for three
reasons: first, it has never been reported---using standard,
unweighted indices this correlation is simply undetectable; 
second, it was known for techniques based on singular
vectors~\cite{LeMSALSATE}; third, we know \emph{exactly} the cause of this
correlation, because the only real difference between
harmonic and closeness centrality is the score assigned to
nodes unreachable from the giant component. We thus expect to discover an unsuspected correlation
between the way PageRank and closeness rank these nodes.

\begin{figure}[htb]
\centering
\includegraphics[scale=.75]{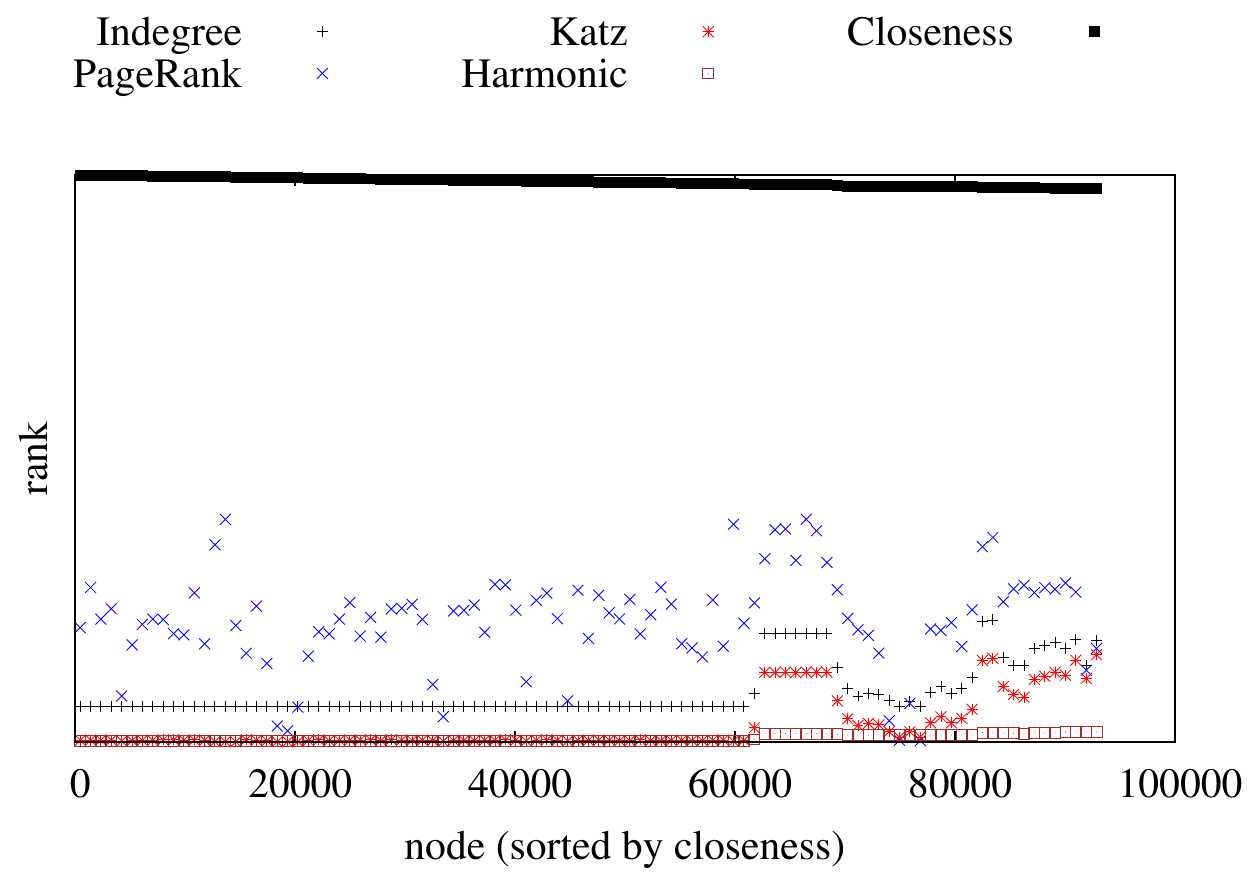}
\caption{\label{fig:wikipediaunreach}Ranks of
components unreachable from the giant component of the Wikipedia graph.}
\end{figure}

To have a visual understanding of what is happening, we created
Figure~\ref{fig:wikipediaunreach}, \ref{fig:hollywoodunreach}
and~\ref{fig:ccunreach} in the following way: first, we isolated the nodes that
are unreachable from the giant component (in the case of Hollywood, which is
undirected, these nodes form separate components), omitting nodes which have
indegree zero, modulo loops (as all measures give the lowest score to such
nodes); then, we sorted the nodes in order of decreasing closeness
rank, and plotted for each node its rank following the other measures (we average ranks on
block of nodes so to contain the number of points in the plots). A point of high
abscissa in the figures implies a high rank.

All three pictures show clearly that \emph{PageRank assigns a preposterously
high rank to to nodes belonging to components that are unreachable from the
giant component}.
This behavior is actually related to PageRank's \emph{insensitivity to size}:
for instance, in a graph made of two components, one of which is a 3-clique and
the other a $k$-clique, the PageRank score of all nodes is  $1/(3+k)$,
independently of $k$. This explains why small dense
components end up being so highly ranked. The same phenomenon is at
work when the community around Lloyd Kaufman's production company (very small and very dense) 
is attributed such a great importance that its elements
make their way to the very top ranks (even 
if Kaufman himself has indegree rank 219 and Debbie Rochon 1790).

We remark that the gap in rank is lower in the case of Wikipedia, but this is
fully in concordance with the higher baseline value of Kendall's $\tau$.

\section{Conclusions}

In this paper, motivated by the need to understand similarity between 
score vectors, such as those generated by centrality measures on large graphs,
we have defined a weighted version of Kendall's $\tau$ starting from its 1945 definition
for scores with ties. We have developed the mathematical properties of our
generalization following a mathematical similarity with internal products, and
showing that for a wide range of weighting schemes our new measure behaves as
expected, providing a correlation index between -1 and 1, and hitting boundaries
only for opposite or equivalent scores.

\begin{figure}[htb]
\centering
\includegraphics[scale=.75]{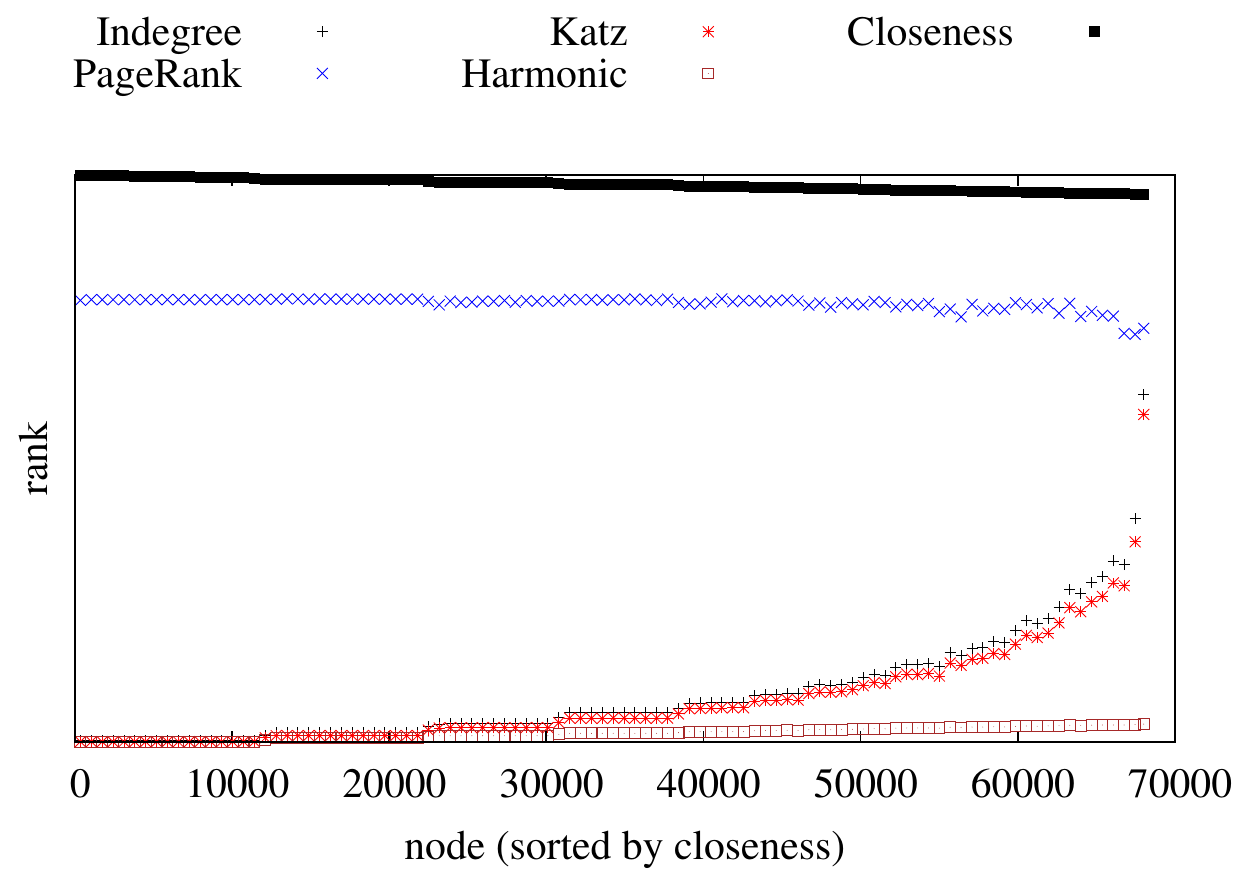}
\caption{\label{fig:hollywoodunreach}Ranks of
components of the Hollywood co-starship graph, excluding the giant component.}
\end{figure}

\begin{figure}[htb]
\centering
\includegraphics[scale=.75]{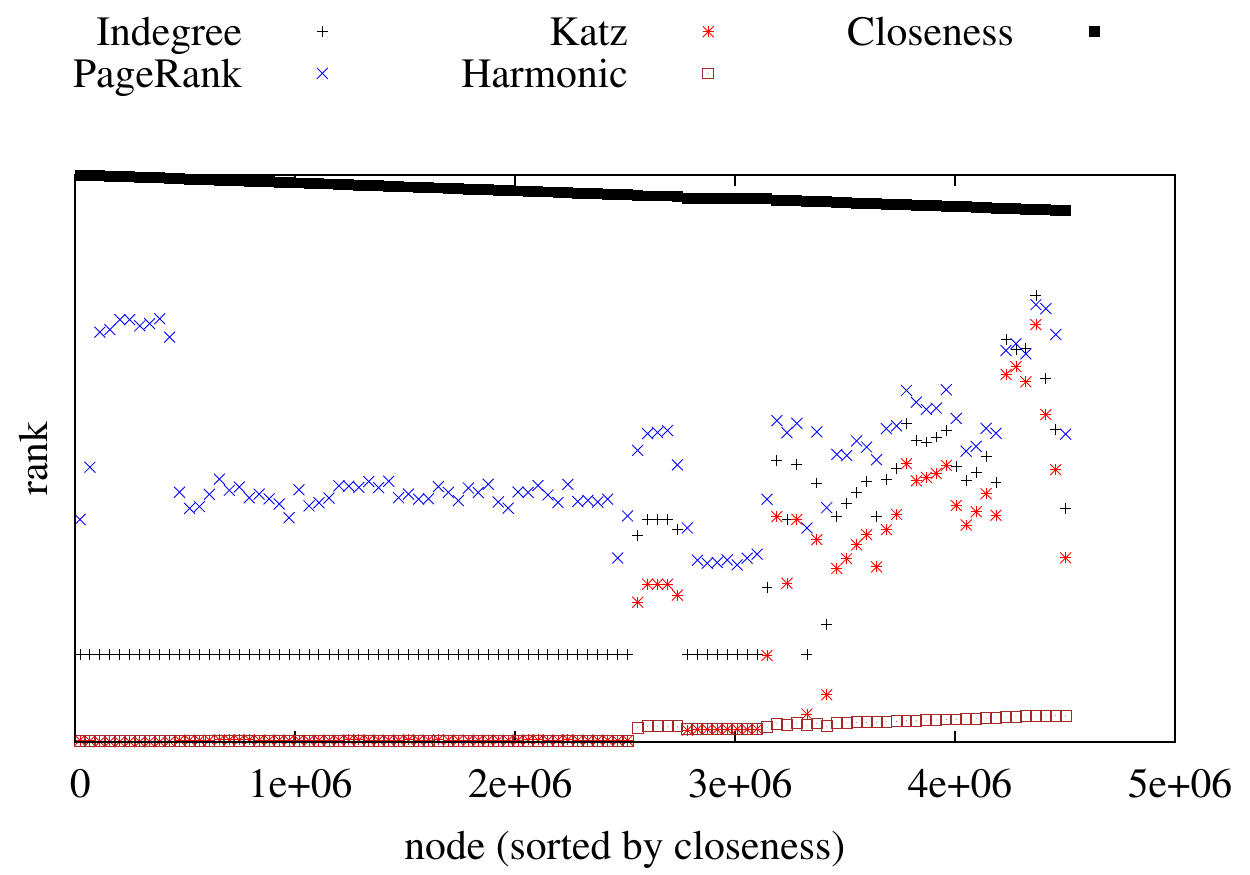}
\caption{\label{fig:ccunreach}Ranks of
components unreachable from the giant component of the Common Crawl host graph.}
\end{figure}

We have then proposed families of weighting schemes that are intuitively
appealing, and showed that they can be computed in time $O(n \log n)$ using 
a generalization of Knight's algorithm, which makes them suitable for
large-scale applications. The fact that the main cost of the algorithm is a
modified stable sort makes it possible to apply standard techniques to run
the algorithm exploiting multicore parallelism, or in distributed environment
such as MapReduce~\cite{DeGMR}. The algorithm can be also used to compute AP correlation~\cite{YARNRCCIR}.

In search for a confirmation of our mathematical intuition, we have then 
applied our measure of choice $\tauh$ (which uses additive hyperbolic
weights) to diverse graph such as Wikipedia, the Hollywood co-starship graph and a large host graph,
finding that, contrarily to Kendall's $\tau$, $\tauh$ provides results that are
consistent with an anecdotal examination of lists of top elements.

Our measure was also able to discover a previously unnoticed correlation
between PageRank and closeness on small components that are unreachable from the giant
component, providing a quantifiable account of the strong bias of PageRank
towards small-sized dense communities. This bias might well be the cause of the 
repeatedly assessed better performance of indegree w.r.t.~PageRank in ranking
documents~\cite{NZTHW,CHUPFFPI}, as in all our experiments the $\tauh$
between PageRank and indegree is above $0.9$.

A generalization similar to the one described in this paper can be also applied
to \emph{Goodman--Kruskal's $\gamma$}, which in the notation of
Section~\ref{sec:computing} is just $(C-D)/(C+D)$. The problem with $\gamma$ is
that the ranking of ties is only implicit (they are simply not counted).
Thus, the value of $w$ on tied pairs does not appear at all in the above
formula. This ``forgetful'' behavior can lead to unnatural results, and suggests
the Kendall's $\tau$ is a better candidate for this approach.

We remark that an interesting application of additive hyperbolic weighting
is that of measuring the correlation between top $k$ lists. By assuming that the rank
function $\rho$ returns $\infty$ after rank $k$, we obtain a correlation index 
that weighs zero pairs outside the top
$k$, weights only ``by one side'' pairs with just one element outside the top
$k$, and weights fully pairs whose elements are within the top $k$.
Formula~(\ref{eq:symm}) could provide then in principle a finer assessment
than, for instance, the modified Kendall's $\tau$ proposed in~\cite{FKSCTKL},
as the position of each element inside the list, beside the fact that it
appears in the top $k$ or not, would be a source of weight. We leave the
analysis of such a correlation measure for future work.

\hyphenation{ Vi-gna Sa-ba-di-ni Kath-ryn Ker-n-i-ghan Krom-mes Lar-ra-bee
  Pat-rick Port-able Post-Script Pren-tice Rich-ard Richt-er Ro-bert Sha-mos
  Spring-er The-o-dore Uz-ga-lis }


\begin{thebibliography}{10}

\bibitem{BavCPTOG}
Alex Bavelas.
\newblock Communication patterns in task-oriented groups.
\newblock {\em J. Acoust. Soc. Am.}, 22(6):725--730, 1950.

\bibitem{BoVAC}
Paolo Boldi and Sebastiano Vigna.
\newblock Axioms for centrality.
\newblock {\em CoRR}, abs/1308.2140, 2013.
\newblock To appear in \textit{Internet Mathematics}.

\bibitem{CHUPFFPI}
Nick Craswell, David Hawking, and Trystan Upstill.
\newblock {Predicting fame and fortune: {P}age{R}ank or indegree?}
\newblock In {\em In Proceedings of the Australasian Document Computing
  Symposium, ADCS2003}, pages 31--40, 2003.

\bibitem{DanRMCUSP}
Henry~E. Daniels.
\newblock The relation between measures of correlation in the universe of
  sample permutations.
\newblock {\em Biometrika}, 33(2):129--135, 1943.

\bibitem{DeGMR}
Jeffrey Dean and Sanjay Ghemawat.
\newblock {MapReduce}: Simplified data processing on large clusters.
\newblock In {\em OSDI '04: Sixth Symposium on Operating System Design and
  Implementation}, pages 137--150, 2004.

\bibitem{FKSCTKL}
Ronald Fagin, Ravi Kumar, and D.~Sivakumar.
\newblock Comparing top $k$ lists.
\newblock {\em SIAM Journal on Discrete Mathematics}, 17(1):134--160, 2003.

\bibitem{FaMAACDRCA}
F.~Farnoud and O.~Milenkovic.
\newblock An axiomatic approach to constructing distances for rank comparison
  and aggregation.
\newblock {\em IEEE Trans.~on Information Theory}, 60(10):6417--6439, October
  2014.

\bibitem{HSBY2}
Johannes Hoffart, Fabian~M. Suchanek, Klaus Berberich, and Gerhard Weikum.
\newblock {YAGO2}: A spatially and temporally enhanced knowledge base from
  {W}ikipedia.
\newblock {\em Artificial Intelligence}, 194:28--61, 2013.

\bibitem{ImCMTDC}
Ronald~L. Iman and W.~J. Conover.
\newblock A measure of top-down correlation.
\newblock {\em Technometrics}, 29(3):351--357, 1987.

\bibitem{KatNSIDSA}
Leo Katz.
\newblock A new status index derived from sociometric analysis.
\newblock {\em Psychometrika}, 18(1):39--43, 1953.

\bibitem{KemMWN}
John~G. Kemeny.
\newblock Mathematics without numbers.
\newblock {\em Daedalus}, 88(4):577--591, 1959.

\bibitem{KenNMRC}
Maurice~G. Kendall.
\newblock A new measure of rank correlation.
\newblock {\em Biometrika}, 30(1/2):81--93, 1938.

\bibitem{KenTTRP}
Maurice~G. Kendall.
\newblock {The treatment of ties in ranking problems}.
\newblock {\em Biometrika}, 33(3):239--251, 1945.

\bibitem{KnCMCKTUD}
William~R. Knight.
\newblock A computer method for calculating {K}endall's tau with ungrouped
  data.
\newblock {\em Journal of the American Statistical Association},
  61(314):436--439, June 1966.

\bibitem{KnuACPSS}
Donald~E. Knuth.
\newblock {\em Sorting and Searching}, volume~3 of {\em The Art of Computer
  Programming}.
\newblock Ad{\-d}i{\-s}on-Wes{\-l}ey, second edition, 1997.

\bibitem{KuVGDBR}
Ravi Kumar and Sergei Vassilvitskii.
\newblock Generalized distances between rankings.
\newblock In {\em Proceedings of the 19th International Conference on World
  Wide Web}, pages 571--580. ACM, 2010.

\bibitem{LeMSALSATE}
Ronny Lempel and Shlomo Moran.
\newblock The stochastic approach for link-structure analysis ({SALSA}) and the
  {TKC} effect.
\newblock {\em Computer Networks}, 33(1):387--401, 2000.

\bibitem{MVLGSWR}
Robert Meusel, Sebastiano Vigna, Oliver Lehmberg, and Christian Bizer.
\newblock Graph structure in the web --- {R}evisited, or a trick of the heavy
  tail.
\newblock In {\em WWW'14 Companion}, pages 427--432. International World Wide
  Web Conferences Steering Committee, 2014.

\bibitem{MeyMAALA}
Carl~D. Meyer.
\newblock {\em Matrix analysis and applied linear algebra}.
\newblock Society for Industrial and Applied Mathematics, 2000.

\bibitem{NZTHW}
Marc Najork, Hugo Zaragoza, and Michael~J. Taylor.
\newblock {HITS} on the web: how does it compare?
\newblock In Wessel Kraaij, Arjen~P. de~Vries, Charles L.~A. Clarke, Norbert
  Fuhr, and Noriko Kando, editors, {\em {SIGIR} 2007: Proceedings of the 30th
  Annual International {ACM} {SIGIR} Conference on Research and Development in
  Information Retrieval, Amsterdam, The Netherlands, July 23-27, 2007}, pages
  471--478. ACM, 2007.

\bibitem{PBMPCR}
Lawrence Page, Sergey Brin, Rajeev Motwani, and Terry Winograd.
\newblock The {P}age{R}ank citation ranking: Bringing order to the web.
\newblock Technical Report SIDL-WP-1999-0120, Stanford Digital Library
  Technologies Project, Stanford University, 1998.

\bibitem{SavCTROS}
I.~Richard Savage.
\newblock Contributions to the theory of rank order statistics---the two-sample
  case.
\newblock {\em The Annals of Mathematical Statistics}, 27(3):590--615, 1956.

\bibitem{ShiWKTS}
Grace~S. Shieh.
\newblock A weighted kendall's tau statistic.
\newblock {\em Statistics \& Probability Letters}, 39(1):17--24, 1998.

\bibitem{SpePMATT}
Charles Spearman.
\newblock The proof and measurement of association between two things.
\newblock {\em The American journal of psychology}, 15(1):72--101, 1904.

\bibitem{VooEHRD}
Ellen~M. Voorhees.
\newblock Evaluation by highly relevant documents.
\newblock In {\em Proceedings of the 24th annual international ACM SIGIR
  conference on Research and development in information retrieval}, pages
  74--82. ACM, 2001.

\bibitem{WMZSMIR}
William Webber, Alistair Moffat, and Justin Zobel.
\newblock A similarity measure for indefinite rankings.
\newblock {\em ACM Trans. Inf. Syst}, 28(4):20:1--20:38, 2010.

\bibitem{YARNRCCIR}
Emine Yilmaz, Javed~A. Aslam, and Stephen Robertson.
\newblock A new rank correlation coefficient for information retrieval.
\newblock In {\em Proceedings of the 31st annual international ACM SIGIR
  conference on Research and development in information retrieval}, pages
  587--594. ACM, 2008.

\end{thebibliography}
\end{document}